\documentclass[reqno,12pt,letterpaper]{amsart}
\usepackage{amsmath,amssymb,amsthm,graphicx,mathrsfs,url}
\usepackage[usenames,dvipsnames]{color}
\usepackage[colorlinks=true,linkcolor=Red,citecolor=Green]{hyperref}
\usepackage{amsxtra}
\usepackage{epstopdf}
\usepackage{subcaption}
\newlength\figureheight 
    \newlength\figurewidth 

\usepackage{graphicx}
\newcommand\smallO{
  \mathchoice
    {{\scriptstyle\mathcal{O}}}
    {{\scriptstyle\mathcal{O}}}
    {{\scriptscriptstyle\mathcal{O}}}
    {\scalebox{.7}{$\scriptscriptstyle\mathcal{O}$}}
  }
\def\?[#1]{\textbf{[#1]}\marginpar{\Large{\textbf{??}}}}

\def\smallsection#1{\smallskip\noindent\textbf{#1}.}
\let\epsilon=\varepsilon 

\newcommand{\RR}{{\mathbb R}}

\newtheorem{theo}{Theorem}
\newtheorem{prop}{Proposition}[section]	
\newtheorem{defi}[prop]{Definition}

\newtheorem{lemm}[prop]{Lemma}
\newtheorem{corr}[prop]{Corollary}
\newtheorem{rem}{Remark}
\newtheorem{ex}{Example}
\numberwithin{equation}{section}

\DeclareMathOperator{\Spec}{Spec}

\DeclareMathOperator{\sgn}{sgn}

\def\indic{\operatorname{1\hskip-2.75pt\relax l}}

\title[Spectral gap in mean-field $\mathcal O(n)$-model]{Spectral gap in mean-field $\mathcal O(n)$-model}
\author{Simon Becker}
\email{simon.becker@damtp.cam.ac.uk}
\address{DAMTP, University of Cambridge, Wilberforce Rd, Cambridge CB3 0WA, UK}
\author{Angeliki Menegaki}
\email{angeliki.menegaki@dpmms.cam.ac.uk}
\address{DPMMS, University of Cambridge, Wilberforce Rd, Cambridge CB3 0WA, UK}

\begin{document}
\maketitle
\begin{abstract}
We study the dependence of the spectral gap for the generator of the Ginzburg-Landau dynamics for all \emph{$\mathcal O(n)$-models} with mean-field interaction and magnetic field, below and at the critical temperature on the number $N$ of particles. For our analysis of the Gibbs measure, we use a one-step renormalization approach and semiclassical methods to study the eigenvalue-spacing of an auxiliary Schr\"odinger operator.
\end{abstract}
\tableofcontents
\section{Introduction and main results}

\subsection{$\mathcal O(n)$-model} The model we are concerned with in this article is the generator of the Ginzburg-Landau dynamics, or Langevin dynamics, of the mean-field $\mathcal O(n)$-model in the critical and supercritical regime $\beta \ge n$, as defined precisely in Section \ref{sec:On}. Our objective is to study the scaling of the spectral gap in terms of the system size $N$, for all the numbers of components $n\ge1$, and including the cases with or without external magnetic field, in the low temperature and critical regime, extending the study of the subcritical regime $\beta < n$ in \cite{BB18}. When $\beta<n$, the spectral gap of the generator remains open uniformly in $N$ and for any number of components $n,$ in the full temperature range. 

The mean-field $\mathcal O(n)$-model is defined by the energy function 
\begin{equation}
\begin{split}
\label{eq:H}
H(\sigma) &= \tfrac{1}{2} \sum_{x \in [N]}\sigma(x)(-\Delta_{\text{MF}}  \sigma)(x) - \frac{1}{\beta} \sum_{x \in [N]} \langle h, \sigma (x) \rangle 
\end{split}
\end{equation}
acting on spin configurations $\sigma: \{1,..,N\} \rightarrow \mathbb S^{n-1}$ where $\Delta_{\text{MF}}$ is the mean-field Laplacian and $h \in \mathbb R^n$ an external magnetic field. For our study of spectral gaps, we consider the Ginzburg-Landau dynamics associated with the Gibbs measure $d\rho \propto e^{-\beta H(\sigma)}$ with Hamilton function \eqref{eq:H}. The inverse temperature parameter $\beta$ is such that lower temperatures (higher $\beta$) favors alignment of spins. The study of mean-field $\mathcal O(n)$-models is motivated by the fact that their behavior approximates that of the full $\mathcal O(n)$-model on high-dimensional tori \cite{ellis1985entropy,Levin_glauberdynamics}.

\subsection{State of the art and motivation} The study of spectral gaps in $\mathcal O(n)$ models is a popular problem that has received a lot of attention over the last decades. The study of logarithmic Sobolev (and other functional) inequalities is a classical and very effective tool to study concentration of measures and to quantify the relaxation rates, \textit{i.e.} the mixing properties, of the dynamics. In particular, the spectral gap (the speed of relaxation) is determined by the constant in the Log-Sobolev inequalities. We define the spectral gap to be the size of the gap between $0$ and the rest of the spectrum of the associated generator $L$, defined in \eqref{eq:Langevin}. The gap then can be also characterized by 
\begin{align} \label{def:spectral_gap}
\lambda_S:= \inf_{f \in L^2(d\rho)\setminus\{0\}} \frac{-\langle Lf,f\rangle_{L^2(d\rho)}}{\operatorname{Var}_\rho(f)}
\end{align}
 where $\operatorname{Var}_\rho$ is the variance relative to the equilibrium measure $\rho$. All these quantities will be specified for our setting in the following section. 
For further background on functional inequalities see \cite{Gross93,BakEm83,Le99,Le01,GZ03,A} and references therein.  

There are only few general approaches for the study of spectral gaps of spin systems, using log-Sobolev inequalities, available and many of them rely on an asymptotic study of log-Sobolev inequalities \cite{LY93,SZ92,StZe92,StrZeg92} or \cite{menz2013} for a more recent result in that direction. In the article \cite{BB18}, a simpler proof for a log-Sobolev inequality was provided for bounded and unbounded spin systems and sufficiently high temperatures. The novelty of the approach in \cite{BB18} is the combination of the study of log-Sobolev inequalities with a simple renormalization group approach to decompose the stationary measure in a way that makes it accessible to simple Bakry-\'Emery techniques. 

Inspired by the method in \cite{BB18}, we invoke the same one-step renormalization group procedure to reduce the high-dimensional problem to the study of a low-dimensional \emph{renormalized measure} and a \emph{fluctuation measure}. In the subcritical regime $\beta < n$, which is the regime analyzed in \cite{BB18},  the renormalization of the equilibrium measure is particularly efficient, since the renormalized potential is strictly convex such that the Bakry-\'Emery criterion can be directly applied to this measure and implies that the spectral gap remains open. This renormalization group method has recently also been successfully applied in the study of the spectral gap for hierarchical spin models \cite{BB18a} and for a lattice discretization of a massive Sine-Gordon model \cite{BB20}.

The low temperature regime, which is the regime we are concerned about within this article, has a non-convex renormalized potential. In this regime, after a single renormalization step, the renormalized potential is not convex. This makes the asymptotic analysis much more difficult and requires new methods:
\medskip

While we analyze the \emph{Ising model}, $n=1$, without magnetic field, directly using explicit criteria for spectral gap and log-Sobolev inequalities \cite{BG,BGL14}, we heavily use the equivalence between the generator of the Ginzburg-Landau dynamics and a Schr\"odinger operator to analyze multi-component,$n \ge 2$, $\mathcal O(n)$-models. This analysis builds heavily upon ideas by Simon \cite{S,CFKS} and Helffer--Sj\"ostrand \cite{Helffer1,Helffer2} who developed effective semiclassical methods to study the low-lying spectrum of Schr\"odinger operators in the semiclassical limit (which in our case corresponds to $N \rightarrow \infty$). These results are discussed thoroughly in the final chapters of \cite{nier2005hypoelliptic}. In this article however, we have to study the spectrum of Schr\"odinger operators beyond the harmonic approximation. In this case, the limiting operator is not explicitly diagonalizable anymore and the spacing between eigenvalues is no longer linear in the semiclassical parameter $N,$ the number of spins. 

The mixing time of the Glauber dynamics of the mean-field Ising model $(\mathcal O(1))$ \emph{without magnetic field} has been carefully analyzed in \cite{DLP,meanfield}. There it is shown-among others- that the mixing time in the subcritical regime $\beta<1$ is $N \log (N)$, the scaling at the critical point $N^{3/2}$ for $\beta=1$ and in the supercritical regime $\beta>1$ it is exponential growing in $N$. This is to be compared to a spectral gap that remains open for $\beta<1$, closes like $N^{-1/2}$ for $\beta=1$ and closes exponentially fast also for $\beta>1$. Thus, the mixing time for the Glauber dynamics are -up to a factor $1/N$- comparable to our findings on the spectral gap, cf. Theorem \ref{theo:theo1}.
\medskip

Our main result on the mean-field \emph{Ising} model in the supercritical regime $\beta>n$ is stated in the following Theorem:
\begin{theo}[Spectral gap--Supercritical Mean-field Ising models, $\beta>1$]
\label{theo:theo1}
Let $N$ be the number of spins and $n$ the number of components. \newline
For the  \emph{supercritical \emph{mean-field Ising model }($n=1, \beta>1$)}, the spectral gap $\lambda_N$ of the generator
\begin{itemize}
\item for the case of small magnetic fields $\vert h \vert < h_{\text{c}}$, closes as $N \rightarrow \infty$ exponentially fast, $\lambda_N =e^{-N\Delta_{\rm{small}}(V)(1+\smallO(1))}$. In particular, for magnetic fields $h \in [0,h_{\text{c}})$ 
\[\Delta_{\rm{small}}(V) = \int_{\gamma_1(\beta)}^{\gamma_2(\beta)} \beta \left(\varphi-\operatorname{tanh}(\beta \varphi+h) \right) d \varphi \] 
where $\gamma_1(\beta) \le  \gamma_2(\beta) \in \RR$ are the two smallest numbers satisfying the condition
\[ \gamma(\beta) = \tanh(\gamma(\beta) \ \beta+h). \] 
\item For critical magnetic fields $\vert h\vert=h_{\text{c}}$, the spectral gap does not close faster than $\Theta(N^{-1/3})$ anymore. 
\item  Finally, for strong magnetic fields $\vert h\vert> h_{\text{c}}$, it is bounded away from zero uniformly in $N$.
\end{itemize} 
where $h_{c}=\sqrt{\beta(\beta-1)}- \operatorname{arccosh}(\sqrt{\beta}).$
\end{theo}

In the case of supercritical multi-component systems ($n\ge 2, \beta>n$) without magnetic fields, it is the rotational invariance of the model that leads to a closing spectral gap as $N$ tends to infinity. To capture this property, we call a function $f:(\mathbb S^n)^N \rightarrow \RR$ \emph{radial}, if it is only a function of the norm of the mean spin $\vert \bar \sigma \vert.$ 
Our main results for all multi-component systems in the supercritical regime $\beta>n$ are summarized in the following Theorem: 
\begin{theo}[Spectral gap--Supercritical Mean-field $\mathcal O(n)$-models, $\beta>n\ge 2$]
\label{theo:2}
Let $N$ be the number of spins and $n$ the number of components. \newline
For the  \emph{supercritical \emph{mean-field $\mathcal O(n)$-models}($n \ge 2, \beta>n$)}, the spectral gap $\lambda_N$ of the generator 
\begin{itemize} 
\item closes as $\lambda_N=\Theta(N^{-1})$ if there is no external magnetic field $h=0$, but remains open $\lambda_N= \Theta(1)$ for radial functions.
\item is bounded away from zero uniformly in the number of spins for all $h \in \mathbb R^n\backslash\{0\}.$ \end{itemize}
\end{theo} 

We also analyze the behavior of the spectral gap at the critical point $\beta=n$ and $h=0$. Using a discrete Fourier analysis approach implemented in Section \ref{sec:critreg} for the Ising case $n=1$ and a direct asymptotic analysis for all higher component systems $n \ge 2$, we find a different asymptotic of the spectral gap from both the supercritical $\beta>n$ (exponentially fast closing) and subcritical $\beta<n$ (spectral gap remains open) regimes:
\begin{theo}[Spectral gap--Critical Mean-field $\mathcal O(n)$ models, $\beta=n$]
\label{theo:theo2}
For all critical, $\beta=n$, $h=0$ mean-field $\mathcal O(n)$-models the spectral gap closes as $\lambda_N=\Theta(N^{-1/2})$. In particular, the rate $N^{-1/2}$ is attained for the magnetization 
\[ M(\sigma) = N^{-1/2} \sum_{x \in [N]} \sigma(x).\]
\end{theo}

We emphasize that at the critical points ($\beta=n$, $h=0$), the gap does no longer close once a non-zero magnetic field is present: 
\begin{theo}[Spectral gap--Mean-field $\mathcal O(n)$ models, $\beta=n$, $h \neq 0$ ]
\label{theo:theo3}
For all, $\beta=n$ and $h \neq 0$, the spectral gap of all mean-field $\mathcal O(n)$-models remains open.
\end{theo}
The proof of Theorem \ref{theo:theo3} is along the lines of Theorem \ref{theo:theo1} in the regime $h > h_{\text{c}}$ and follows from Proposition \ref{strongfields} in the Ising-case, $n=1$, and in the multi-component case, $n \ge 2$, from Proposition \ref{Simon}.

\subsection{Organization of the article}
The article is organized as follows:
\begin{itemize}
\item In Section \ref{sec:On} we introduce the mean-field $\mathcal O(n)$-model.
\item In Section \ref{sec:renormalized} we introduce the renormalized methods.
\item In Section \ref{sec:Ising} we analyze the mean-field Ising model in the supercritical regime $\beta>1$ and prove Theorem \ref{theo:theo1}.
\item In Section \ref{sec:High-dim} we analyze the higher-component mean-field $\mathcal O(n)$-models in the supercritical regime $\beta>n$ and prove Theorem \ref{theo:2}.
\item In Section \ref{sec:critreg} we study the critical regime and prove both Theorems \ref{theo:theo2} and \ref{theo:theo3}.
\item Our article contains an appendix that contains technical details and further details on numerical methods.
\end{itemize}

\smallsection{Acknowledgements} 
Both authors were supported by the EPSRC grant EP/L016516/1 for the University of Cambridge CDT, the CCA. The authors are grateful to Roland Bauerschmidt for many useful discussions and bringing this problem to our attention.

\smallsection{Notation}
We write $f(z) = \mathcal O(g(z)) $ to indicate that there is $C>0$ such that $\left\lvert f(z) \right\rvert \le C \left\lvert g(z) \right\rvert$ and $f(z)= \smallO(g(z))$ for $z \rightarrow z_0$ if there is for any $\varepsilon>0$ a neighbourhood $U_{\varepsilon}$ of $z_0$ such that $\left\lvert f(z)\right\rvert \le \varepsilon \left\lvert g(z) \right\rvert.$ We say that $f(z) = \Theta(g(z))$ if there are $ k_1, k_2>0$ and $z_0$ such that for all $z \ge z_0$  we have $k_1  g(z) \leq f(z) \leq k_2  g(z).$
The expectation with respect to a measure $\mu$ is written denoted by $\mathbb E_{\mu}(X).$
The normalized surface measure on the $n$ sphere is denoted as $dS_{\mathbb S^{n}}.$ We write $\indic $ to denote a vector or matrix whose entries are all equal to one and $\operatorname{id}$ for the identity map. Finally, we introduce the notation $[N]:=\left\{1,...,N \right\}.$ The eigenvalues of a self-adjoint matrix $A$ shall be denoted by $\lambda_1(A)\le ... \le \lambda_N(A)$.

\section{The mean-field $\mathcal O(n)$-model}
\label{sec:On}
We study the mean-field $\mathcal O(n)$-model with spin configuration $\sigma: [N] \rightarrow \mathbb S^{n-1}$ and introduce the mean-field Laplacian $(\Delta_{\text{MF}}\sigma)(x):= \tfrac{1}{N}\sum_{y \in [N]} \left( \sigma(y)-\sigma(x)\right).$

The mean spin is defined as $\overline{\sigma}:=\tfrac{1}{N} \sum_{x \in [N]} \sigma(x).$
The energy of a spin configuration $\sigma$ is given by the \emph{Curie-Weiss Hamiltonian}
\begin{equation}
\begin{split}
\label{eq:Hamenergy}
H(\sigma) &= \tfrac{1}{2} \sum_{x \in [N]}\sigma(x)(-\Delta_{\text{MF}}  \sigma)(x) - \frac{1}{\beta} \sum_{x \in [N]} \langle h, \sigma (x) \rangle \\
& = \tfrac{1}{4N} \sum_{x,y \in[N]} \left\lvert \sigma(x)-\sigma(y) \right\rvert^2 - \frac{1}{\beta} \sum_{x \in [N]} \langle h, \sigma (x) \rangle \\
&=\tfrac{N}{2}(1-\left\lvert \overline{\sigma} \right\rvert^2) - \tfrac{N}{\beta}  \  \langle h,  \overline{\sigma}  \rangle.
\end{split}
\end{equation}
where the constant vector $h \in \mathbb R^n$ represents an external magnetic field and $\beta$ is the inverse temperature of the system. The critical temperature for the $\mathcal O(n)$-models is $\beta=n$ and we study both regimes: the supercritical regime $\beta>n $ and the critical regime $\beta=n$.

The dynamics we consider is the continuous-time Ginzburg-Landau dynamics
\begin{equation}
\begin{split}
\label{eq:krasmo}
\partial_t f &= \sum_{x \in [N]} \left\langle \nabla^{(x)}_{\mathbb S^{n-1}}, \beta^{-1} \nabla^{(x)}_{\mathbb S^{n-1}}f+ f \nabla^{(x)}_{\mathbb S^{n-1}} H \right\rangle_{\mathbb R^n}
\end{split}
\end{equation}
to the invariant distribution of the mean-field $\mathcal O(n)$-model which is the Gibbs measure $d\rho(\sigma):=e^{- \beta H(\sigma)} /Z \ dS^{\otimes N}_{\mathbb S^{n-1}}(\sigma)$ with normalizing constant $Z$. The operators $\Delta^{(x)}_{\mathbb S^{n-1}}$ defined by $\langle f,-\Delta^{(x)}_{\mathbb S^{n-1}}f \rangle:=\langle \nabla^{(x)}_{\mathbb S^{n-1}}f,\nabla^{(x)}_{\mathbb S^{n-1}}f\rangle$ and $\nabla^{(x)}_{\mathbb S^{n-1}}$ are the Laplace-Beltrami and gradient operator on $\mathbb S^{n-1}$ acting on spin $i$, respectively. We recall that for the Ising model $n=0$ and a function $F:\mathbb S^0 \rightarrow \mathbb R$, the gradient is given by $(\nabla_{\mathbb S^0} F)(\sigma)=F(\sigma)-F(-\sigma).$
The $L^2\left((\mathbb S^{n-1})^N\right)$-adjoint of the generator of the Kramers-Smoluchowski equation \eqref{eq:krasmo} is the generator
\begin{equation}
\begin{split}
\label{eq:Langevin}
(L\zeta)(\sigma)&:= \sum_{x \in [N]}\beta^{-1}(\Delta^{(x)}_{\mathbb S^{n-1}}\zeta)(\sigma)-\langle (\nabla^{(x)}_{\mathbb S^{n-1}} H)(\sigma), (\nabla^{(x)}_{\mathbb S^{n-1}}\zeta)(\sigma) \rangle_{\mathbb R^n}. 
\end{split}
\end{equation}
Studying the operator $L$ on the weighted space $L^2\left((\mathbb S^{n-1})^N, \ d\rho \right)$ makes this generator self-adjoint.
The quadratic form of the generator \eqref{eq:Langevin} is just a rescaled Dirichlet form
\begin{equation*}
\begin{split}
-\langle Lf,f\rangle_{L^2(d\rho)} & = \beta^{-1} \sum_{x \in [N]} \left\lVert  \nabla_{\mathbb S^{n-1}}^{(x)}f \right\rVert^2_{L^2(d\rho)}.
\end{split}
\end{equation*}

\section{Renormalized measure and mathematical preliminaries}
\label{sec:renormalized}
We start with the definition of entropy with respect to probability measures:
\begin{defi}[Entropy]
For a probability measure $\mu $ on some Borel set $\Omega$ the entropy $\operatorname{Ent}_{\mu}(F)$ of a positive measurable function $F:\Omega \rightarrow \mathbb R_{\ge 0}$ with $\int_{\Omega} F(x) \log^+(F(x)) \ d\mu(x)< \infty$ is defined as 
\begin{equation}
\begin{split}
\operatorname{Ent}_{\mu}(F) &:=\int_{\Omega} F(x) \log\left(F(x) \bigg/  \int_{\Omega} F(y) \ d\mu(y) \right) \ d\mu(x).
\end{split}
\end{equation}
\end{defi}

Instead of studying the generator of the dynamics directly, we apply a one step renormalization first \cite[Sec.\@ $1.4$]{BBS18}:

\begin{defi}[Renormalized quantities]
The \emph{renormalized single spin potential} $V_n$ associated with the mean-field $ \mathcal O(n)$-model for $\varphi \in \mathbb R^n$ is defined as 
\begin{equation}
\begin{split}
\label{eq:renpot}
V_n(\varphi) &= -\log \int_{\mathbb S^{n-1}} e^{-\tfrac{\beta}{2}\left\lVert \varphi-\sigma \right\rVert^2+ \left\langle h,  \sigma \right\rangle}\ dS_{\mathbb S^{n-1}}(\sigma) \\
&= \tfrac{\beta}{2}\left(1+\left\lVert \varphi \right\rVert^2 \right) - \log \left( \Gamma\left(\tfrac{n}{2}\right) \left(\tfrac{2}{ \left\lVert \beta \varphi + h \right\rVert}\right)^{\frac{n}{2}-1} I_{\tfrac{n}{2}-1}( \left\lVert \beta \varphi +h \right\rVert) \right)
\end{split}
\end{equation}
where $I$ is the modified Bessel function of the first kind.
The \emph{$N$-particle renormalized measure} is defined for a normalizing constant $\nu_N^{(n)}$ by 
\begin{equation}
\label{eq:renmeasure}
d\nu_{N}(\varphi) = \nu_N^{(n)}e^{-NV_n(\varphi)} \ d\varphi\text{ on }\mathbb R^n.
\end{equation}
\end{defi}
\begin{defi}[Fluctuation measure]
For any $\varphi \in \mathbb R^n,$ there is a probability measure $\mu_{\varphi}$,  the \emph{fluctuation measure}, on $\left(\mathbb S^{n-1}\right)^N$ defined as 
\begin{equation}
\label{eq:flucmeasure}
 \mathbb E_{\mu_{\varphi}}(F) = \int_{\left(\mathbb S^{n-1}\right)^N} F(\sigma)  e^{NV_n(\varphi)}  \prod_{x \in [N]} e^{-\tfrac{\beta}{2}\left\lVert \varphi-\sigma(x) \right\rVert_2^2+ \left\langle h, \sigma(x) \right\rangle} \ dS(\sigma(x)).
 \end{equation}
 \end{defi}
A straightforward calculation shows that the stationary measure $d\rho$ can be decomposed into the fluctuation and renormalized measure such that $\mathbb E_{\rho}(F)= \mathbb E_{\nu_N}(\mathbb E_{\mu_{\varphi}}(F)).$

\begin{ex}
In the case of the \emph{Ising model} $(n=1)$ the renormalized potential is 
\begin{equation}
\label{eq:Ising}
 V_1(\varphi) = \tfrac{\beta}{2}(1+  \varphi^2)  - \log\left(\cosh(\beta \varphi+h)\right).
 \end{equation}
For the \emph{XY model} $(n=2)$ the renormalized potential reads
\begin{equation}
\label{eq:XY} V_2(\varphi) = \tfrac{\beta}{2}(1+  \left\Vert\varphi\right\Vert^2)  - \log\left(I_0(\left\lVert \beta \varphi+h \right\rVert)\right)
 \end{equation}
where $I$ is the modified Bessel function of the first kind.

For the \emph{Heisenberg model} $(n=3)$ one finds
\[ V_3(\varphi)= \tfrac{\beta}{2}\left(1+ \left\lVert \varphi \right\rVert^2\right)- \log\left( \frac{\sinh\left\lVert \beta \varphi  +h\right\rVert}{\left\lVert \beta \varphi+h  \right\rVert}\right).\]
\end{ex}

For the $N$-asymptotic study of eigenvalues we observe that the renormalized potential grows quadratically at infinity such that $\Delta V_n \in L^{\infty}(\mathbb R^n).$

The Ginzburg-Landau dynamics dynamics for the renormalized measure is then given by the self-adjoint operator $L_{\text{ren}}:D(L_{\text{ren}}) \subset L^2(\mathbb R^n, d\nu_N) \rightarrow L^2(\mathbb R^n, d\nu_N)$, satisfying
\begin{equation}
\label{eq:renLang}
(L_{\text{ren}} \zeta)(\varphi) =  (\Delta_{\mathbb R^n} \zeta)(\varphi) - N\left \langle \nabla_{\mathbb R^n} V_n(\varphi), \nabla_{\mathbb R^n} \zeta (\varphi)\right\rangle. 
\end{equation}
The renormalized generator $L_{\text{ren}}$ satisfies 
\begin{equation}
\label{eq:renmeasuprop}-\left\langle L_{\text{ren}} f,f \right\rangle_{L^2(d\nu_N)}=\left\lVert \nabla_{\mathbb R^n}f \right\rVert^2_{L^2(d\nu_N)}.
\end{equation}

The renormalized Schr\"odinger operator with null space spanned by $e^{-NV_n} $ is the operator defined by conjugation $- \Delta_{\operatorname{ren}} = e^{-NV_n/2} L_{\text{ren}}  e^{NV_n/2}$
\begin{equation}
\label{eq:renWitt}
\Delta_{\text{ren}} =- \Delta_{\mathbb R^n}+\tfrac{N^2}{4}\vert \nabla V_n(\varphi)\vert^2- \tfrac{N}{2}\Delta V_n(\varphi).
\end{equation}

\begin{defi}[LSI and SGI]
Let $\mu$ be a Borel probability measure on $\mathbb R^n$. We say that $\mu$ satisfies a logarithmic Sobolev inequality \emph{LSI(k)} iff
\[ \operatorname{Ent}_{\mu}(f^2) \le \tfrac{2}{k} \left\lVert \nabla f \right\rVert_{L^2(d \mu)}^2\]
for all smooth functions $f$. 
The \emph{LSI(k)} implies \cite[Prop.\ $2.1$]{Le99} that $\mu$ satisfies a spectral gap inequality \emph{SGI(k)}
\[ \operatorname{Var}_{\mu}(f) \le \tfrac{1}{k}  \left\lVert \nabla f \right\rVert_{L^2(d \mu)}^2.\]
\end{defi}
Thus, in light of the characterisation \eqref{def:spectral_gap}, the spectral gap of $L_{\text{ren}}$ is by \eqref{eq:renmeasuprop} precisely the constant in the SGI of the renormalized measure.
\begin{rem}
If $f$ vanishes outside a set $\Omega$ of measure $\mu(\Omega) <1$ and if $\mu$ satisfies a \emph{SGI(k)} then
\begin{equation}
\label{eq:CS}
\left\lVert f \right\rVert_{L^2(d \mu)}^2\le \tfrac{1}{k(1-\mu(\Omega))}  \left\lVert \nabla f \right\rVert_{L^2(d \mu)}^2.
\end{equation}
\end{rem}

For Borel probability measures $\mu$ on $\mathbb R$ there is an explicit characterization of the measures satisfying a {LSI} \cite[Theorem $5.3$]{BG}:

Any such measure $\mu$ satisfies a {LSI(k)} iff there exist absolute constants $K_0=1/150$ and $K_1=468$ such that the optimal value $k$ in the {LSI(k)} satisfies $K_0(D_0+D_1)\le 1/k \le K_1(D_0+D_1)$ for finite $D_0$ and $D_1$. Let $m$ be the median of $\mu$ and $p(t) \ dt$ the absolutely continuous part of $\mu$ with respect to Lebesgue measure. The constants $D_0$ and $D_1$ are given by
\begin{equation}
\begin{split}
\label{eq:LSIconstants}
D_0 &:=\sup_{x<m} \left( -\mu((-\infty,x])\log(\mu((-\infty,x]))  \int_x^m \frac{ds}{p(s)}\right)\text{ and } \\ 
D_1 &:=\sup_{x>m} \left( -\mu([x,\infty))\log(\mu([x,\infty))  \int_m^x \frac{ds}{p(s)}\right).
\end{split}
\end{equation}

For constants 
\begin{equation}
\begin{split} 
\label{eq:SGIconstants}
B_0 &:=\sup_{x<m} \left( \mu((-\infty,x])  \int_x^m \frac{ds}{p(s)}\right) \text{ and }
B_1 :=\sup_{x>m} \left( \mu([x,\infty)) \int_m^x \frac{ds}{p(s)}\right)
\end{split}
\end{equation}
one defines the Muckenhoupt number \cite{M72} $B:=\operatorname{max}(B_0,B_1).$ The measure $\mu$ satisfies then a SGI with optimal constant $c=1/k$ if and only if $B$ is finite in which case 
\begin{equation}
\label{eq:SGI}
 B/2  \le c \le 4 B
\end{equation}
\cite[Theorem $4.5.1$]{BGL14}. 

\begin{rem}
\label{redrem}
The proof given in \cite[Theorem $5.3$]{BG} shows that the characterization of \emph{LSI} constants holds true not only by splitting at the median: Instead, there is $\varepsilon>0$ such that for any $\zeta$ for which $\mu((-\infty,\zeta]), \mu([\zeta,\infty)) \in (1/2-\varepsilon,1/2+\varepsilon)$ the above characterization \eqref{eq:LSIconstants} holds true when the median $m$ is replaced by $\zeta.$ The same is, up to an unimportant adaptation of the lower bound in \eqref{eq:SGI}, for the \emph{SGI} as well, cf.\@ \cite[Prop.\@ $3.2$ + $3.3$]{GR01}. 
\end{rem}
We continue by observing that the fluctuation measures satisfy a LSI($\tfrac{2}{\gamma_n}$) independent of $h$ or $\varphi$. This follows for $n=1$ with $\gamma_n = 4$ from a simple application of the tensorization principle to the classical bound on the Bernoulli distribution \cite{A, Le01, SC}. For number of components $n \ge 2$ one can use the results from \cite{ZQM}.
\begin{prop}
\label{theo1}
 Let the renormalized measure $\nu_N$ satisfy a \emph{LSI($\lambda$)}, then the full equilibrium measure $\rho$ satisfies a \emph{LSI} 
\[  \operatorname{Ent}_{\rho}(F^2) \le \tfrac{2}{\gamma_n} \left(1+\tfrac{8 N  \beta^2}{\lambda} \right) \sum_{x \in [N]} \left\lVert  \nabla_{\mathbb S^{n-1}}^{(x)}F \right\rVert^2_{L^2(d\rho)}\]
and if the renormalized measure $\nu_N$ satisfies a \emph{SGI($\lambda$)}, then the equilibrium measure $\rho$ satisfies a \emph{SGI}
\[  \operatorname{Var}_{\rho}(F) \le \tfrac{1}{\gamma_n} \left(1+\tfrac{4N \beta^2}{\lambda} \right)\sum_{x \in [N]} \left\lVert  \nabla_{\mathbb S^{n-1}}^{(x)}F \right\rVert^2_{L^2(d\rho)}.\]
\end{prop}
\begin{proof}
The proof of the SGI is as follows:
For the SGI we obtain the decomposition
\begin{equation}
\label{eq:SGIstep}
\begin{split}
\operatorname{Var}_{\rho}(F)&=\mathbb E_{\nu_N}( \operatorname{Var}_{\mu_{\varphi}}(F))+ \operatorname{Var}_{\nu_{N}}(\mathbb E_{\mu_{\varphi}}(F)) \\
&\le  \tfrac{1}{\gamma_n} \sum_{x \in [N]} \left\lVert  \nabla_{\mathbb S^{n-1}}^{(x)}F \right\rVert^2_{L^2(d\rho)} + \tfrac{1}{\lambda} \mathbb E_{\nu_N}\left(\left\lvert \nabla_{\varphi} \mathbb E_{\mu_{\varphi}}(F) \right\rvert^2\right).
\end{split}
\end{equation}
To bound the second term in the above estimate, we compute using the Cauchy-Schwarz inequality and the spectral gap inequality for fluctuation measures $\mu_{\varphi}$ on the sphere, defined by \eqref{eq:flucmeasure} such that, see \cite[Theorem $1$, (11)-(15)]{BB18},
\begin{equation}
\begin{split}
\label{eq:equation}
\nabla_{\varphi} \mathbb E_{\mu_{\varphi}}(F) &=N \nabla V(\varphi) \mathbb E_{\mu_{\varphi}}(F) -\beta \sum_{x \in [N]}\mathbb E_{\mu_{\varphi}}(F(\varphi-\sigma_x)).
\end{split}
\end{equation}

We then use that by the explicit expression \eqref{eq:renpot}
\begin{equation}
\begin{split}
 \nabla V(\varphi) =\beta \frac{\int_{\mathbb S^{n-1}} e^{-\frac{\beta}{2} \Vert \varphi-\sigma_x \Vert^2} (\varphi-\sigma_x)dS_{\mathbb S^{n-1}}(\sigma_x)
}{\int_{\mathbb S^{n-1}} e^{-\frac{\beta}{2} \Vert \varphi-\sigma_x \Vert^2} \ dS_{\mathbb S^{n-1}}(\sigma_x)} =\beta( \varphi- \mathbb E_{\mu_{\varphi}}(\sigma_x)).
\end{split}
\end{equation}
Inserting this into \eqref{eq:equation} we find that 
\[ \nabla_{\varphi} \mathbb E_{\mu_{\varphi}}(F)  = \beta  \sum_{x \in [N]}  \mathbb E_{\mu_{\varphi}}(F\sigma_x)- \mathbb E_{\mu_{\varphi}}(F) \mathbb E_{\mu_{\varphi}}(\sigma_x) = \beta \sum_{x \in [N]} \operatorname{cov}_{\mu_{\varphi}}(F, \sigma_x).\]
Thus, we have using Cauchy-Schwarz that 
\begin{equation}
\label{eq:Cauchy-Schwarz}
 \left\lvert \nabla_{\varphi} \mathbb E_{\mu_{\varphi}}(F) \right\rvert^2  \le N  \beta^2 \sum_{x \in [N]} \vert \operatorname{cov}_{\mu_{\varphi}}(F, \sigma_x) \vert^2.
 \end{equation}
We can then use that by Cauchy-Schwarz again 
\begin{equation}
\begin{split}
 \operatorname{cov}_{\mu_{\varphi}}(F, \sigma_x)&=\mathbb E_{\mu_{\varphi}}\left((F-\mathbb E_{\mu_{\varphi}}(F))(\sigma -\mathbb E_{\mu_{\varphi}}(\sigma)) \right) \\
&\le \sqrt{\mathbb E_{\mu_{\varphi}}(F-\mathbb E_{\mu_{\varphi}}(F))^2}  \sqrt{\mathbb E_{\mu_{\varphi}}(F-\mathbb E_{\mu_{\varphi}}(F))^2} \sqrt{\mathbb E_{\mu_{\varphi}}(\sigma-\mathbb E_{\mu_{\varphi}}(\sigma))^2}\\
&\le 2\sqrt{\operatorname{Var}_{\mu_{\varphi}}(F)}.
\end{split}
\end{equation}
Finally, inserting this into 
\eqref{eq:Cauchy-Schwarz} and using the LSI for the fluctuation measure, we find 
\begin{equation}
\begin{split}
 \mathbb E_{\nu_N}\left(\left\lvert \nabla_{\varphi} \mathbb E_{\mu_{\varphi}}(F) \right\rvert^2\right) &\le 4N \beta^2 \mathbb E_{\nu_N}\operatorname{Var}_{\mu_{\varphi}}(F) \\
 &\le \frac{4N \beta^2}{\gamma_n} \mathbb E_{\rho}\sum_{x \in [N]} \left\vert \nabla^{(x)}_{\mathbb S^{n-1}} F \right\rvert^2
 \end{split}
\end{equation}
which after inserting this bound into \eqref{eq:SGIstep} implies the claim. To prove the LSI we follow \cite{BB18} and write \begin{align*} 
\operatorname{Ent}_{\rho}(F^2) &= \mathbb{E}_{\nu_N} \left(  \operatorname{Ent}_{\mu_\varphi}(F^2)  \right) + \operatorname{Ent}_{\nu_N}\left(\mathbb{E}_{\mu_\varphi}( F^2 ) \right)  \\ &\le 
\tfrac{2}{\gamma_n} \sum_{x \in [N]} \left\lVert  \nabla_{\mathbb S^{n-1}}^{(x)}F \right\rVert^2_{L^2(d\rho)} + \tfrac{2}{\lambda} \mathbb E_{\nu_N}\left(  \left\lvert \nabla_{\varphi} \sqrt{ \mathbb E_{\mu_{\varphi}}(F^2) } \right\rvert^2 \right).
\end{align*} 
For the second term we have from applying the Cauchy-Schwarz inequality: \begin{align*} 
\left\lvert \nabla_{\varphi} \sqrt{ \mathbb E_{\mu_{\varphi}}(F^2) } \right\rvert^2  = \beta^2 \left\lvert \frac{\sum_{x \in [N]}  \operatorname{cov}_{\mu_{\varphi}}(F^2, \sigma_x) }{ \sqrt{ \mathbb E_{\mu_{\varphi}}(F^2) } }   \right\rvert^2 \le \beta^2N \frac{\sum_{x \in [N]} \vert \operatorname{cov}_{\mu_{\varphi}}(F^2, \sigma_x)  \vert^2}{\mathbb E_{\mu_{\varphi}}(F^2)}.
\end{align*} 

By \textit{ doubling the variables} $\sigma_x, \sigma_x'$, we write 
\begin{align*} 
\vert \operatorname{cov}_{\mu_\varphi}(F^2(\sigma_x), \sigma_x) \vert &= \frac{1}{2} \left\vert \mathbb{E}_{\mu_\varphi}\Big( (F^2(\sigma_x)-F^2(\sigma_x')  )(\sigma_x-\sigma_x') \Big)  \right\vert \\ &\le \sqrt{\operatorname{Var}_{\mu_\varphi}(F)} \sqrt{  \frac{1}{2} \mathbb{E}_{\mu_\varphi \otimes \mu_\varphi} \left(  ( F(\sigma_x)+F(\sigma_x'))^2(\sigma_x-\sigma_x')^2\right)   } \\ &\le 
\sqrt{\operatorname{Var}_{\mu_\varphi}(F)} \sqrt{8\mathbb{E}_{\mu_\varphi}(F^2)   }
\end{align*} 
where in the last two lines we applied CS inequality and used that $ | \sigma_x - \sigma_x'|\le 2$. Then  $$ \vert \operatorname{cov}_{\mu_\varphi}(F^2, \sigma)  \vert^2 \le 8 \operatorname{Var}_{\mu_\varphi}(F)  \mathbb{E}_{\mu_\varphi}(F^2).$$ 
This gives
\begin{align*}
\left\lvert \nabla_{\varphi} \sqrt{ \mathbb E_{\mu_{\varphi}}(F^2) } \right\rvert^2 \le 8 N \beta^2 \operatorname{Var}_{\mu_\varphi}(F). 
\end{align*}
Overall we have 
 \begin{align*}
\operatorname{Ent}_{\rho}(F^2) &\le \tfrac{2}{\gamma_n} \sum_{x \in [N]} \left\lVert  \nabla_{\mathbb S^{n-1}}^{(x)}F \right\rVert^2_{L^2(d\rho)} + \tfrac{16 \beta^2 N }{\lambda \gamma_n} \sum_{x \in [N]} \left\lVert  \nabla_{\mathbb S^{n-1}}^{(x)}F \right\rVert^2_{L^2(d\rho)} \\ & = \tfrac{2}{\gamma_n}\left(1+ \tfrac{8 \beta^2 N}{\lambda} \right)  \sum_{x \in [N]} \left\lVert  \nabla_{\mathbb S^{n-1}}^{(x)}F \right\rVert^2_{L^2(d\rho)}.
 \end{align*} 
 
\end{proof}

\section{The mean-field Ising model}
\label{sec:Ising}
Without loss of generality, we assume $h \ge 0$ when studying the Ising model.
We define the critical magnetic field strength in the Ising model 
\[ h_{c}(\beta):=\sqrt{\beta(\beta-1)}- \operatorname{arccosh}(\sqrt{\beta})\] for temperatures $\beta\ge 1$ as the supremum of all $h>0$ such that $x = \tanh(\beta x+h)$ has three distinct solutions for $x \in [-1,1].$ In particular $h_c(\beta)$ is monotone with respect to the inverse temperature $\beta$. 

\medskip

The critical magnetic field strength is chosen in such a way that for fields $h<h_c(\beta)$ there are two potential wells in the renormalized potential landscape, see Figure \ref{Figure1}, whereas for $h \ge h_c(\beta)$ there is only one, see Figure \ref{Figure2} in subsection \ref{sec:SGSMF} where this case is discussed.

\subsection{Lower bound on spectral gap in weak field $h< h_c(\beta)$ regime}
We start by showing that the inverse spectral gap in the Ising model in the case of subcritical magnetic fields, i.e. $h<h_c(\beta),$ converges at most exponentially fast to zero as the number of spins, $N$, increases. 

\medskip

We start by showing a LSI with exponential constant for the renormalized measure. This implies by Prop. \ref{theo1} that such an LSI must also hold for the full many-particle measure $d\rho.$
\begin{figure}
\centerline{\includegraphics[height=7cm]{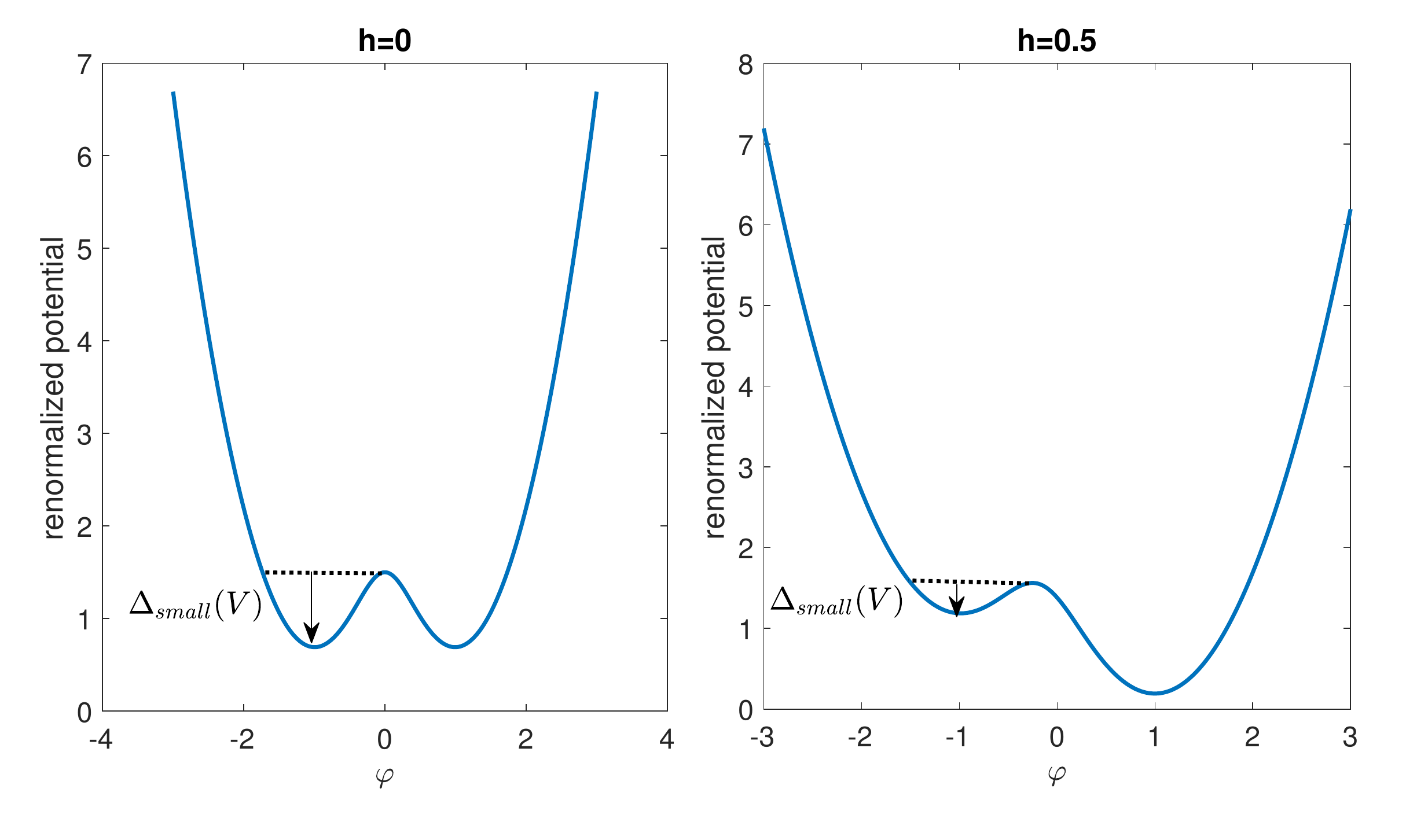}} 
\caption{\emph{Weak magnetic fields:} Renormalized potentials for the Ising model with $\beta=3$ and zero $h=0$ or weak $h=0.5$ magnetic fields form a double well.}
\label{Figure1}
\end{figure}
\begin{prop}[LSI for $\nu_N$]
\label{weakone}
Let $\beta>1$ and $h<h_c(\beta)$ such that $V_1$ is a double well potential where the depth of the smaller well is denoted by $\Delta_{\rm{small}}(V),$ cf. Fig.\@ \ref{Figure1}. The mean-field Ising model satisfies a \emph{LSI}$\left(e^{-N\Delta_{\rm{small}}(V)(1+\smallO(1))}\right)$\footnote{If the magnetic field is zero, i.e. $h=0,$ both wells are of equal size.}
\[ \operatorname{Ent}_{\nu_N}(F^2) \lesssim  e^{N\Delta_{\rm{small}}(V)(1+\smallO(1))} \int_{\mathbb R} \left\lvert F' \right\rvert^2 \ d\nu_N. \]
\end{prop}
\begin{proof}
The renormalized potential $V_1$ has on $[0,\infty)$ a global minimum with positive second derivative at some $\varphi_{\text{min}}$ satisfying $\varphi_{\text{min}} = \tanh(\beta \varphi_{\text{min}}+h).$ This follows since the renormalized potential \eqref{eq:renpot} reduces to
\[V_1(\varphi) = \frac{\beta}{2} \varphi^2-\log(\cosh(\beta \varphi+h)) \] 
and the critical points of this potential are easily found to satisfy $\varphi= \tanh(\beta \varphi+h),$ see also \cite[Lemma $1.4.6$]{BBS18}.
For small temperatures, i.e. $\beta \rightarrow \infty,$ one has $\varphi_{\text{min}}=1+\smallO(1).$

\emph{We first consider $h=0:$}
In this case, the median of the renormalized measure is located precisely at $\varphi=0$ and $\varphi_{\text{min}}>0$ is one of the two non-degenerate global minima of the renormalized potential (the other minimum is located at $-\varphi_{\text{min}}$ by axisymmetry).

An application of Laplace's principle, see \cite[Ch.\@ II,Theorem $1$]{W01}, shows that for all $x>0:$
\begin{equation}
\begin{split}
\label{eq:h=0}
&\lim_{N \rightarrow \infty}\tfrac{1}{N} \log\left(-\nu_N( [x,\infty)) \log(\nu_N( [x,\infty))) \int_0^x \frac{e^{N V_1(\varphi)} }{\nu_N^{(1)}}\ d\varphi \right) \\
&=\lim_{N \rightarrow \infty} \tfrac{1}{N}\left( \log \int_{x}^{\infty} e^{-NV_1(\varphi) } \ d\varphi  + \log\left(- \log(\nu_N( [x,\infty)))\right)+ \log\int_0^x e^{N V_1(\varphi )} \ d\varphi   \right) \\
&= -\inf_{t \in [x,\infty)} V_1(t)  +\sup_{t \in [0,x]} V_1(t).
\end{split}
\end{equation}
The supremum of \eqref{eq:h=0} is attained at $x=\varphi_{\text{min}}$ such that 
\[-\inf_{t \in [x, \infty)} V_1(t)  +\sup_{t \in [0,x]} V_1(t)=\Delta_{\rm{small}}(V)>0.\]

Here, we used that for $x> \varphi_{\text{min}}$ we get by Laplace's principle 
\[-\log\left(\nu_N [x,\infty) \right) = N(V(x)-V(\varphi_{\text{min}}))(1+\smallO(1))\] 
and thus $\lim_{N \rightarrow \infty} \tfrac{\log\left(-\log\left(\nu_N [x,\infty) \right)\right)}{N}=0.$ \\
 On the other hand, if $x \in (0,\varphi_{\text{min}})$ then, again by Laplace's principle, $-\log\left(\nu_N [x,\infty) \right) = -\log(\tfrac{1}{2})+\smallO(1)$ and thus $\lim_{N \rightarrow \infty} \tfrac{\log\left(-\log\left(\nu_N [x,\infty) \right)\right)}{N}=0$ as well. The case $x=\varphi_{\text{min}}$ can be treated analogously.
Hence, we obtain for the constant $D_1$ as in \eqref{eq:LSIconstants} 
\begin{equation}
\begin{split}
\label{eq:D1}
D_1&:=\sup_{x>0}\left(-\nu_N( [x,\infty)) \log(\nu_N( [x,\infty))) \int_0^x \frac{e^{N V_1(\varphi)} }{\nu_N^{(1)}} \ d\varphi \right)  \\
&= e^{N\Delta_{\rm{small}}(V)(1+\smallO(1))}.
\end{split}
\end{equation}
The symmetry of the distribution for $h=0$ implies then that $D_0=D_1$.

\emph{We now consider $h>0$:} The renormalized potential possesses a unique global minimum at some $\varphi_{\text{min}}$ and the median of the renormalized measure converges to this point $\varphi_{\text{min}}$, see Fig.\@ \ref{Figure1}, as Laplace's principle implies 
\[\frac{\int_{\varphi_{\text{min}}}^{\infty} e^{-N V_n(\varphi)} \ d\varphi}{\int_{-\infty}^{\infty} e^{-N V_n(\varphi)} \ d\varphi}=\frac{1}{2}+\mathcal O(1/N).\]
Hence, it suffices to verify the {LSI} bounds \eqref{eq:LSIconstants}  for $m=\varphi_{\text{min}}$ as argued in Remark \ref{redrem}.

Arguing as in \eqref{eq:D1} yields for $h>0$ and $x< \varphi_{\text{min}}:$
\begin{equation}
\begin{split}
\label{eq:hneq0}
&\lim_{N \rightarrow \infty}\tfrac{1}{N} \log\left(-\nu_N( (-\infty,x] \log(\nu_N((-\infty,x] )) \int^{m}_x \frac{e^{N V_1(\varphi)} }{\nu_N^{(1)}} \ d\varphi \right) \\
&=\lim_{N \rightarrow \infty} \tfrac{1}{N}\left( \log \int_{-\infty}^{x} e^{-NV_1(\varphi )} \ d\varphi   + \log\left(- \log(\nu_N( [x,\infty))) \right)+ \log\int_x^m e^{N V_1(\varphi )} \ d\varphi  \right) \\
&= -\inf_{t \in (-\infty,x]} V_1(t)  +\sup_{t \in [x,m]} V_1(t)
\end{split}
\end{equation}
which shows $D_0= e^{N\Delta_{\rm{small}}(V)(1+\smallO(1))}$ by taking $x$ to be the minimum of the smaller well of the renormalized potential.
For the constant $D_1$ we get on the other hand for $x>\varphi_{\text{min}}$, since the renormalized potential is monotonically increasing on $[\varphi_{\text{min}},\infty),$ 
\begin{equation}
\begin{split}
&\lim_{N \rightarrow \infty}\tfrac{1}{N} \log\left(-\nu_N( [x,\infty)) \log(\nu_N( [x,\infty)) \int_{m}^x \frac{e^{N V_1(\varphi)} }{\nu_N^{(1)}} \ d\varphi \right) \\
&=\lim_{N \rightarrow \infty} \tfrac{1}{N}\left( \log \int_{x}^{\infty}  e^{-NV_1(\varphi )} \ d\varphi   + \log\int_{m}^x e^{N V_1(\varphi )} \ d\varphi  + \log\left(- \log(\nu_N( [x,\infty)))\right) \right) \\
&= -\inf_{t \in [x,\infty)} V_1(t)  +\sup_{t \in [m,x]} V_1(t) = 0
\end{split}
\end{equation}
such that $D_1$ is negligible compared with $D_0.$
\end{proof}

\subsection{Upper bound on spectral gap in weak field $h<h_c(\beta)$ regime.}
The upper bound on the spectral gap is obtained by finding an explicit trial function saturating the SGI. For this construction, we use the notation and results of Lemma \ref{lem:asym}.

\medskip 

In order to fix ideas first, we assume $h=0$. We start by observing that the mean spin $\overline{\sigma}$ can only take values in the set $\mathcal M:=\left\{-1, -1+2/N,...,1 \right\}.$ The weights of the stationary measure $d\rho$ are given by functions $\eta_N:\mathcal M \rightarrow \mathbb R$
\begin{equation}
\label{eq:etaN}
\eta_N(i) :=\sum_{\sigma;\overline{\sigma}=i} e^{-\beta H(\sigma)}=\binom{N}{N/2(1+i)} e^{-\tfrac{N\beta}{2}(1-i^2)}.
\end{equation}
where we used \eqref{eq:Hamenergy}.

\medskip 

We also introduce trial functions $f_N: \left\{ \pm 1 \right\}^N \rightarrow \mathbb R$ for the spectral gap inequality given by
\begin{equation}
\label{eq:fN}
f_N(\sigma) := \sum_{i \in \mathcal M; 0 \le i \le \overline{\sigma}}\frac{\indic_{\left\{i \le \gamma_3(\beta)\right\}}}{\eta_N(i)}  
\end{equation}
with indicator function $\indic$ and $\gamma_3(\beta)$ is the largest solution to $\varphi = \tanh(\beta \varphi+h).$
Since $f_N$ depends only on the mean spin, we can identify them with functions $g_N:\mathcal M \rightarrow \mathbb R$
\[g_N(m) =\sum_{i \in \mathcal M; 0 \le i \le m}\frac{\indic_{\left\{i \le \gamma_3(\beta)\right\}}}{\eta_N(i)}  
 \text{ such that }f_N(\sigma)=g_N(\overline{\sigma}).  \]
For the $L^2$ norm of the $f_N$ we find
\begin{equation}
\begin{split}
\label{eq:lowerbound}
 \left\lVert f_N \right\rVert_{L^2(d\rho)}^2
 &= \sum_{i \in \mathcal M} \frac{\eta_N(i)}{Z} \vert g_N(i) \vert^2  \ge  \sum_{i \in \mathcal M; i > \gamma_3(\beta)} \frac{\eta_N(i)}{Z} \left(\sum_{j \in \mathcal M; 0 \le j \le \gamma_3(\beta)}\frac{1}{\eta_N(j)}  \right)^2
 \end{split}
 \end{equation}
where $Z$ is the normalization constant of the full measure $d\rho$.
For the gradient of $f_N$ we find \[\left\vert \nabla_{\mathbb S^0}^{(i)} f_N(\sigma) \right\vert^2  = \left\vert g_N(\overline{\sigma})-g_N(\overline{\sigma} \pm 2/N) \right\vert^2 \lesssim \eta_N(\overline{\sigma})^{-2}.\]
Hence, for some $C>0$
\begin{equation}
\label{eq:upperbd}
 \sum_{i \in [N]} \left\lVert \nabla_{\mathbb S^0}^{(i)} f_N\right\rVert_{L^2(d\rho)}^2  \le  \frac{CN}{Z}\sum_{i \in \mathcal M; 0 \le i \le \gamma_3(\beta)} \frac{1}{\eta_N(i)}.
 \end{equation}
Using \eqref{eq:CS} with $\mu(\Omega) = \tfrac{1}{2}$ implies by comparing \eqref{eq:lowerbound} with \eqref{eq:upperbd} that the constant $\gamma$ in the SGI is bounded from below by 
\begin{equation}
\label{eq:ratio}
\frac{1}{NC} \sum_{i \in \mathcal M; i > \gamma_3(\beta)} \eta_N(i) \ \sum_{i \in \mathcal M; 0\le i \le \gamma_3(\beta)} \frac{1}{\eta_N(i)}  \le  \frac{1}{2\gamma}.
\end{equation}
We recall from the discussion in Lemma \ref{lem:asym} that the continuous approximation $\eta_N(i)$ attains its maximum in the limit at $i= \gamma_3(\beta)$ and the summand $\frac{1}{\eta_N(i)}$ in the second sum attains its maximum in the limit at $i= 0$.

Thus it suffices to study the asymptotic of the logarithm of the leading order summands in \eqref{eq:ratio} using the asymptotic behaviour of $\zeta_{N} := \partial_s \log(\eta_{N}(s))$ given in \eqref{eq:logder}
\begin{equation*}
\begin{split}
\log\left(\tfrac{\eta_N(\gamma_3(\beta))}{\eta_N(0)}\right) &= N \int_0^{\gamma_3(\beta)} \zeta_N(s) \ ds= N \int_0^{\gamma_3(\beta)} (\beta s -\operatorname{arctanh}(s)) \ ds \ (1+\smallO(1)) \\
&= N \left(\frac{\beta \gamma_3(\beta)^2}{2} - \int_0^{\gamma_3(\beta)\beta}  \frac{\operatorname{arctanh}(x/\beta)}{\beta} \ dx \right)(1+\smallO(1))\\
&= -N \int_0^{\gamma_3(\beta)} \beta  (x -\operatorname{tanh}(\beta x)) \ dx \ (1+\smallO(1)) \\
&= N \Delta_{\rm{small}}(V)(1+\smallO(1)).
\end{split}
\end{equation*}
Here, we used integration of the inverse function to obtain the last line and \eqref{eq:Ising} in the last one.
In the case of a positive weak magnetic field $h \in (0, h_c(\beta))$ we choose a trial function $f_{N,h}: \left\{ \pm 1 \right\}^N \rightarrow \mathbb R$ given by
\begin{equation}
\begin{split}
\label{eq:fNh}
f_{N,h}(\sigma) &:= \sum_{i \in \mathcal M;  \overline{\sigma}<i<\gamma_3(\beta) }\frac{\indic_{\{i \ge \gamma_1(\beta)\}}}{\eta_{N,h}(i)} \text{
where for }i \in \mathcal M\\
\eta_{N,h}(i) &:= \binom{N}{N/2(1+i)} e^{-\tfrac{N\beta}{2}\left(1-i^2\right)+h N i }.
\end{split}
\end{equation}
Proceeding as above in \eqref{eq:lowerbound} we obtain for the $L^2$ norm the lower bound
\begin{equation}
\begin{split}
\label{eq:lowerbd2}
 \left\lVert f_{N,h} \right\rVert_{L^2(d\rho)}^2 \ge \frac{1}{Z}\sum_{i \in \mathcal M; i < \gamma_1(\beta)} \eta_{N,h}(i) \left( \sum_{j \in \mathcal M;\gamma_3(\beta)> j \ge \gamma_1(\beta)} \frac{1}{\eta_{N,h}(j)} \right)^2.
 \end{split}
 \end{equation}
For the Dirichlet form we find, as for \eqref{eq:upperbd}, for some $C>0$
\begin{equation}
\label{eq:upperbd2}
 \sum_{x \in [N]} \left\lVert \nabla_{\mathbb S^0}^{(x)} f_{N,h} \right\rVert_{L^2(d\rho)}^2  \le  \frac{CN}{Z}  \sum_{j \in \mathcal M;\gamma_3(\beta)> j \ge \gamma_1(\beta)}  \frac{1}{\eta_{N,h}(j)}.
 \end{equation}
We can apply \eqref{eq:CS} with $\mu(\Omega) = \tfrac{1}{1-\varepsilon}$ for some $\varepsilon>0$ since the trial function \eqref{eq:fNh} vanishes to the right of the global maximum such that by comparing \eqref{eq:lowerbd2} with \eqref{eq:upperbd2}  the constant $\gamma$ in the SGI is bounded from below by 
\begin{equation}
\frac{1}{NC}\sum_{i \in \mathcal M; i < \gamma_1(\beta)}\eta_{N,h}(i) \sum_{j \in \mathcal M;\gamma_3(\beta)> j \ge \gamma_1(\beta)}\frac{1}{\eta_{N,h}(j)}  \le  \frac{1}{\varepsilon \gamma}.
\end{equation}
The weight $\eta_{N,h}(i)$ in the first sum attain their maximum (in the limit) at $i=\gamma_1(\beta)$ and the summands $\frac{1}{\eta_{N,h}(i)}$ in the second sum attain their maximum at $i=\gamma_2(\beta)$.

To explicitly state an upper bound on the spectral gap it suffices to study the asymptotic of the logarithm of the leading order summands
\begin{equation*}
\begin{split}
\log\left(\tfrac{\eta_{N,h}(\gamma_1(\beta))}{\eta_{N,h}(\gamma_2(\beta))}\right) &=N \int_{\gamma_2(\beta)} ^{\gamma_1(\beta)} \zeta_{N,h}(s) \ ds=  \int_{\gamma_2(\beta)}^{\gamma_1(\beta)} (\beta s -\operatorname{arctanh}(s)) \ ds \ (1+\smallO(1)) \\
&= N\left(\frac{\beta (\gamma_1(\beta)^2-\gamma_2(\beta)^2)}{2} - \int_{\gamma_2(\beta)\beta}^{\gamma_1(\beta)\beta}  \frac{\operatorname{arctanh}(x/\beta)}{\beta} \ dx \right)(1+\smallO(1))\\
&= -N \int_{\gamma_2(\beta)}^{\gamma_1(\beta)} \beta  (x -\operatorname{tanh}(\beta x)) \ dx \ (1+\smallO(1)) \\
&= N\Delta_{\rm{small}}(V)(1+\smallO(1)).
\end{split}
\end{equation*}

\subsection{Spectral gap in strong magnetic field regime $h> h_{c}(\beta)$.}
\label{sec:SGSMF}
\begin{figure}
\centerline{\includegraphics[height=7cm]{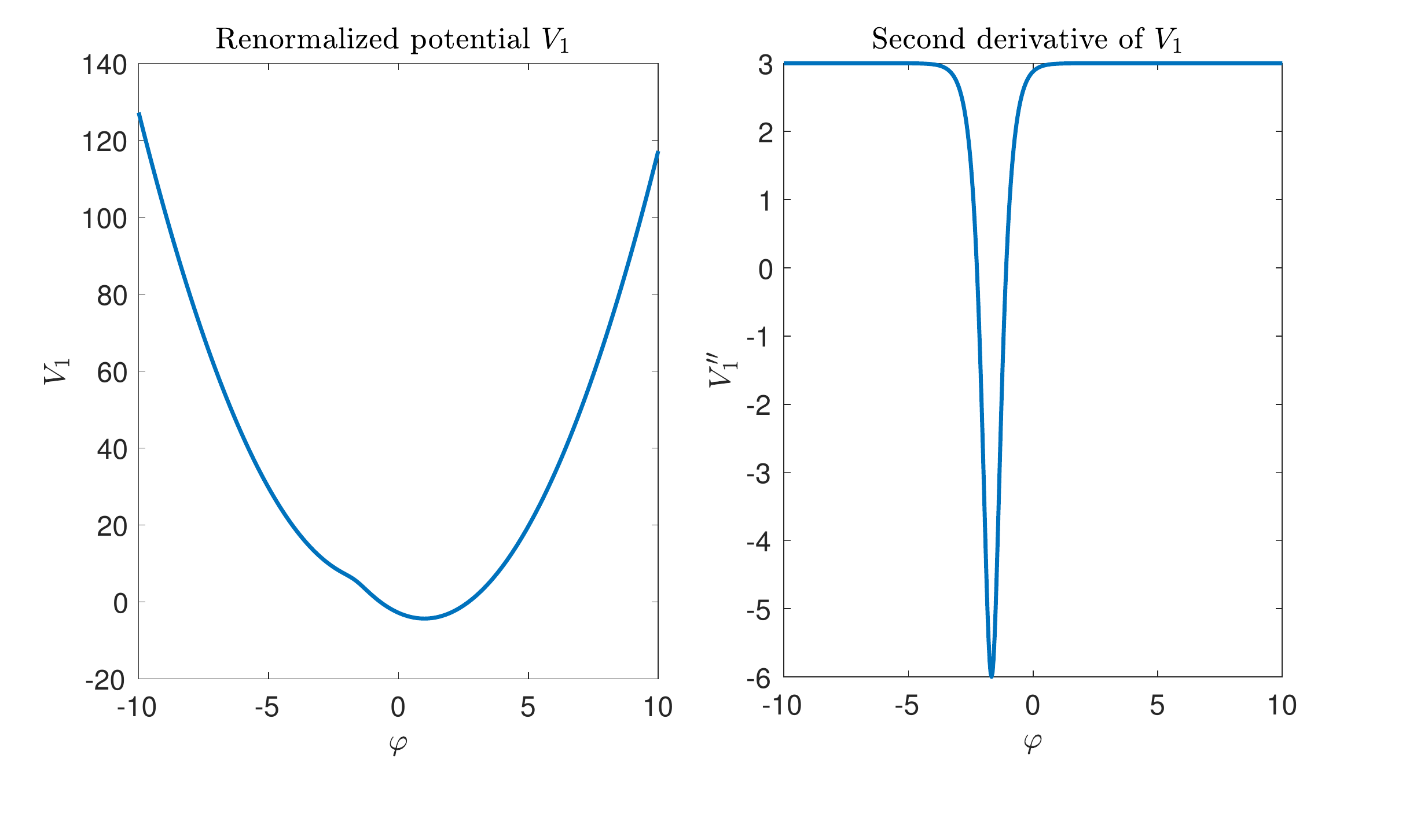}} 
\caption{\emph{Strong magnetic fields:} The renormalized potential and its second derivative for $h=5$ for $\beta=3$. The potential is non-convex even though it is a single well potential. However, it is convex in a neighbourhood of the global minimum.}
\label{Figure2}
\end{figure}
Next, we study the case of strong magnetic fields for the Ising model, that is $V_1'$ has at most one root, for $\beta >1$. We also include the case $\beta=1$ and $h \neq 0$.
Unlike in the case of weak magnetic fields, in which case the constant in the {LSI} for the renormalized measure is exponentially increasing in the number of spins, the spectral gap of the renormalized measure is now linearly increasing in the number of spins. Responsible for this uniform gap is the local uniform convexity at the minimum of the renormalized potential. More precisely, we have
\begin{equation}
\begin{split}
V_1'(\varphi) =\beta (\varphi-\operatorname{tanh}(\beta \varphi +h)) \text{ and } V_1''(\varphi) = \beta(1 - \beta \operatorname{sech}(\beta \varphi +h)^2).
\end{split}
\end{equation}
Thus, $V_1''(\varphi)=0$ yields $\varphi_{\pm}= \frac{-h \pm \operatorname{arccosh}(\sqrt{\beta})}{\beta}.$ Inserting this into $V_1'(\varphi_{\pm})=0$ implies that $h_{\pm}=\pm \operatorname{arccosh}(\sqrt{\beta}) \mp \sqrt{\beta(\beta-1)}$ with sign $\sgn(h_{\pm})=\mp 1$ and thus $\varphi_{\pm}= \pm \sqrt{\frac{\beta-1}{\beta}}.$
In particular, in the subcritical regime $\beta>1$ all global minima $\varphi_{*}$ have sign $\sgn(\varphi_{*})=\sgn(h)$, such that the renormalized potential satisfies $V_1''(\varphi_{*})>0.$ Moreover, for $\beta =1$ and $h \neq 0$ there are no points at which both the first and second derivative vanish. The third derivative at this point however is always non-zero and given by 
\[ V_1^{(3)}(\varphi_{\pm}) = \mp 2 \sqrt{\beta(\beta-1)} \beta. \]
\begin{prop}[Ising model, strong field]
\label{strongfields}
Let $\beta\ge 1$ and $h> h_c(\beta)$, i.e. $V_1$ is a single well potential. We obtain for the Ising model a \emph{SGI$(\gamma)$}
\[ \operatorname{Var}_{\nu_N}(F) \le  \tfrac{1}{\gamma} \int_{\mathbb R} \left\lvert F' \right\rvert^2 \ d\nu_N \] 
where $\frac{1}{\gamma}$ is uniformly bounded in $N.$ 
\end{prop}
\begin{proof}
Since the renormalized Schr\"odinger operator and renormalized generator are unitarily equivalent up to a factor, see \eqref{eq:renWitt}, the semiclassical eigenvalue distribution stated in \cite[Theo. $1.1$]{S} implies the statement of the Proposition: 

\medskip

It follows immediately from the renormalized Schr\"odinger operator \eqref{eq:renWitt}
\begin{equation}
\Delta_{\text{ren}} =- \tfrac{d^2}{d\varphi^2}+\tfrac{N^2}{4}\vert V_1'(\varphi)\vert^2- \tfrac{N}{2} V_1''(\varphi).
\end{equation}
that the low-lying eigenfunction of $\Delta_{\text{ren}}$ accumulate at the unique non-degenerate (the second derivative is non-zero) potential well and the spectral gap of the renormalized measure grows linearly in $N.$ The result then follows from Prop. \ref{theo1}.
\end{proof}

\subsection{Critical magnetic fields in $n=1$}

\begin{figure}
\centerline{\includegraphics[height=7cm]{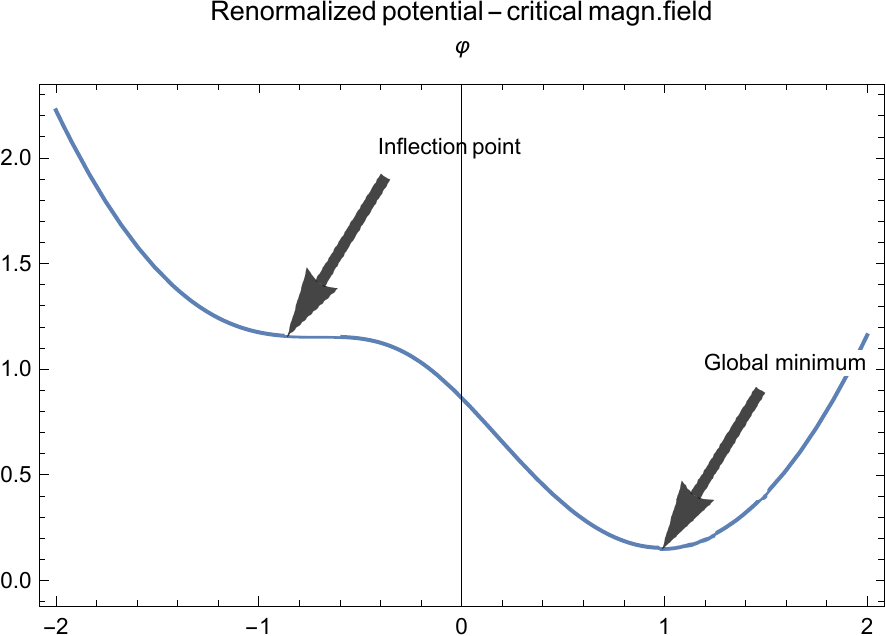}} 
\caption{\emph{Critical magnetic field for $n=1$:} Renormalized potential of the Ising model with $\beta=2$ possesses two critical points, one inflection point and a global minimum.}
\label{fig:crit}
\end{figure}
\begin{prop}
Let $h=h_{\pm}$ and $\beta > 1.$ The spectral gap of the radial renormalized Schr\"odinger operator grows as $\Theta( N^{2/3})$ and in particular, the spectral gap of the full measure does not close faster than $\Theta( N^{-1/3}).$
\end{prop}
\begin{proof}
Let $\lambda:= N/2$ and consider the Schr\"odinger operator, defined in \eqref{eq:renWitt},
\begin{equation}
\begin{split}
\label{eq:Ham1}
\mathcal H&:=-\partial_x^2 + \lambda^2 \vert V'_1(x) \vert^2-\lambda V_1''(x)
\end{split}
\end{equation}
for the renormalized potential and auxiliary Schr\"odinger operators, which are obtained as the Taylor expansion of \eqref{eq:Ham1}
\begin{equation}
\begin{split}
\label{eq:Hamiltonian3}
H_{\varphi_*}&=-\partial_{x}^2 + \lambda^2 \vert V_1'(\varphi_{*}) \vert^2 (x-\varphi_{*})^2 - \lambda V_1''(\varphi_{*}) \text{ and } \\
H_{\varphi_{\pm}}&=-\partial_{x}^2 + \lambda^2 \beta^3(\beta-1) (x-\varphi_{\pm})^4 \pm  \lambda  2\sqrt{\beta(\beta-1)} \beta (x-\varphi_{\pm})
\end{split}
\end{equation}
on $L^2(\RR)$ localized to the two critical points, the inflection point $\varphi_{\pm}$ and the global minimum $\varphi_{*}.$
We then define $j \in C_c^{\infty}((-2,2);[0,1])$ such that $j(x)=1$ for $\left\lvert x\right\rvert \le 1$ and from this functions
\begin{equation}
\begin{split}
\label{eq:property2}
J_{\varphi_{*} }(x) &:= j(\lambda^{2/5}\vert x-\varphi_{*} \vert ), \ J_{\varphi_{\pm}}(x) := j(\lambda^{3/10}\vert x-\varphi_{\pm} \vert )  \text{ and }\\
&J(x) := \sqrt{1-J_{\varphi_{\pm}}(x)^2-J_{\varphi_{*} }(x)^2} \text{ with }\\
&\left\lVert \nabla  J_{\varphi_{*} } \right\rVert_{\mathbb R^n}^2=\mathcal O(\lambda^{4/5}),\left\lVert \nabla J_{\varphi_{\pm}} \right\rVert_{\RR^n}^2 = \mathcal O(\lambda^{3/5}).
\end{split}
\end{equation}
Invoking then unitary maps $U_{\varphi_{*} }, U_{\varphi_{\pm}} \in \mathcal L(L^2(\mathbb R))$ defined as 
\begin{equation}
\begin{split}
 (U_{\varphi_{*} }f)(x) := \lambda^{-1/4}f(\lambda^{-1/2}(x+\varphi_{*})) \text{ and }(U_{\varphi_{\pm} }f)(x) := \lambda^{-1/6}f(\lambda^{-1/3}(x+\varphi_{\pm})) 
\end{split}
\end{equation}
shows that the two Schr\"odinger operators in \eqref{eq:Hamiltonian3} are in fact unitarily equivalent, up to multiplication by powers of $\lambda$, to the $\lambda$-independent Schr\"odinger operators
\begin{equation}
\begin{split}
\label{eq:S}
S_{\varphi_*}&=- \partial_x^2+\vert V_1'(\varphi_{*}) \vert^2 x^2 -V_1''(\varphi_{*}) \\
S_{\varphi_{\pm}} &= - \partial_x^2+\beta^3(\beta-1) x^4 \pm 2\sqrt{\beta(\beta-1)} \beta x,
\end{split}
\end{equation}
respectively. Both operators have discrete spectrum and that $\inf(\Spec(S_{\varphi_{\pm}}))>0$ is shown in Section \ref{sec:SUSYQM}. We illustrate the behaviour of the smallest eigenvalues of $S_{\varphi_{\pm}}$ in Figure \ref{fig:Sop}. More precisely, we have that
\begin{equation}
\begin{split}
\label{eq:uniteqv2}
\lambda  U_{\varphi_*}^{-1}\ S_{\varphi_*} U_{\varphi_*}&= H_{\varphi_*}\text{ and }
\lambda^{2/3}  U_{\varphi_{\pm}}^{-1}\ S_{\varphi_{\pm}} U_{\varphi_{\pm}}= H_{\varphi_{\pm}}.
\end{split}
\end{equation}

Taylor expansion of the potential at the respective critical point and the estimate on the gradient \eqref{eq:property2} imply that  
\begin{equation}
\begin{split}
\left\lvert J_{\varphi_{*} }(\mathcal H-H_{\varphi_{*} })J_{\varphi_{*} } \right\rvert &= \mathcal O(\lambda^{4/5}) \text{ and also }\left\lvert J_{\varphi_{\pm} }(\mathcal H-H_{\varphi_{\pm} })J_{\varphi_{\pm} } \right\rvert = \mathcal O(\lambda^{3/5}). 
\end{split}
\end{equation}
Let $0=e_1<e_2\le..$ be the eigenvalues (counting multiplicities) of $S_{\varphi_{*}}$ and $0<f_1\le f_2\le...$ the ones of $S_{\varphi_{\pm}}$  and choose $\tau$ such that $\lambda e_{n+1}>\tau >\lambda e_n$ and $\lambda^{2/3} f_{m+1}>\tau >\lambda^{2/3} f_m$ with $P_i$ being the projection onto the eigenspace to all eigenvalues of $S_i$ below $\tau.$
\begin{figure}
\centerline{\includegraphics[height=7cm]{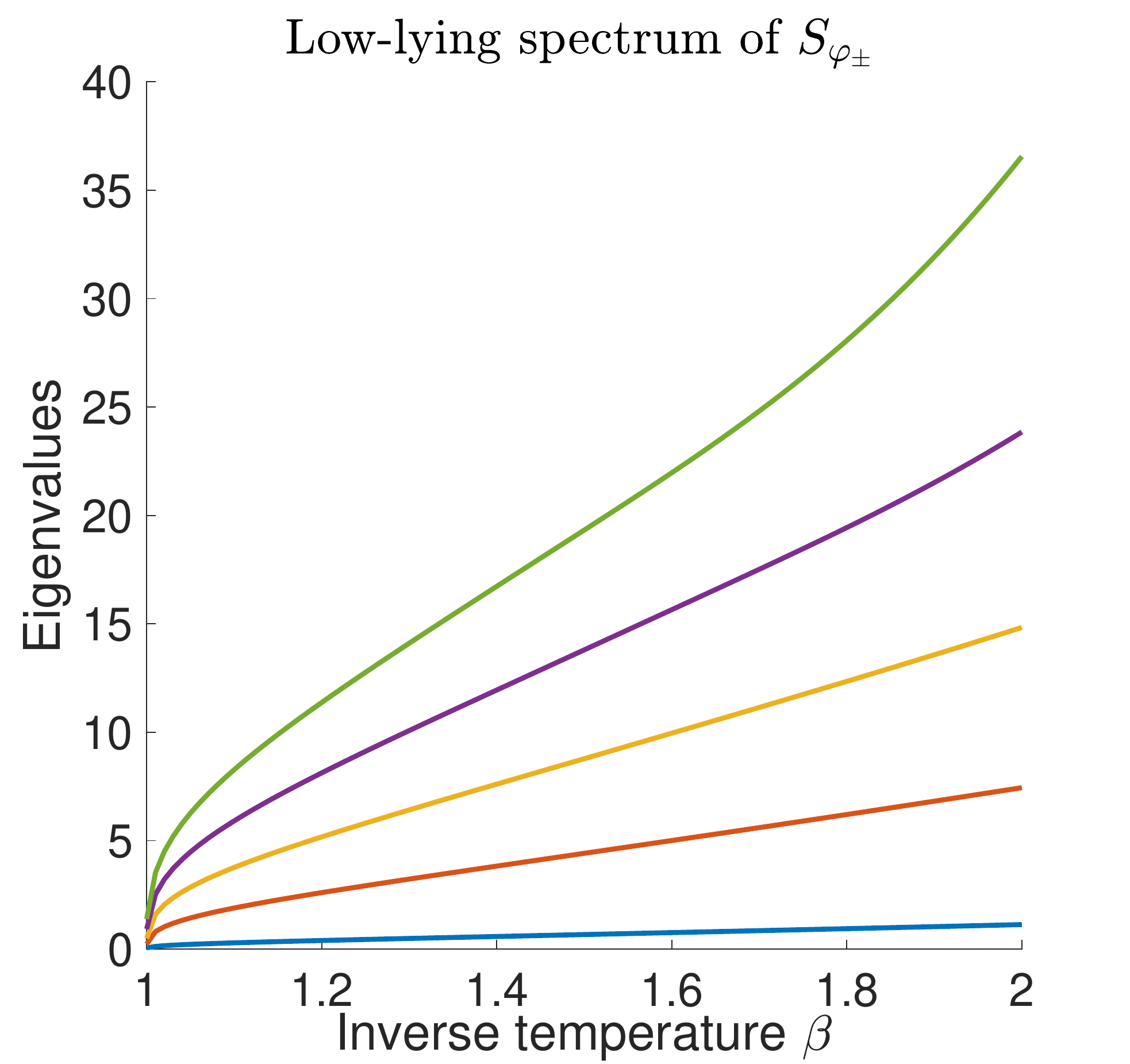}} 
\caption{The five smallest eigenvalues of the operator $S_{\varphi_{\pm}}$ as a function of $\beta.$ The smallest eigenvalue is strictly positive.}
\label{fig:Sop}
\end{figure}
The IMS formula, see \cite[(11.37)]{CFKS} for a version on manifolds, implies that  
\begin{equation}
\label{eq:usefulrepresent3}
\mathcal H = J \mathcal H J -\vert \nabla J \vert^2+\sum_{i \in \{\varphi_{*}, \varphi_{\pm}\}}\left(J_{i} H_i J_{i} + J_{i} (\mathcal H-H_i) J_{i}-  \vert \nabla J_{i} \vert^2 \right).
\end{equation}
On the other hand, it follows that 
\begin{equation*} 
\begin{split}
J_{\varphi_{*}} H_{\varphi_{*}} J_{\varphi_{*}}
&= J_{\varphi_{*}} H_{\varphi_{*}} P_{\varphi_{*}} J_{\varphi_{*}} + J_{\varphi_{*}} H_{\varphi_{*}} (\operatorname{id}-P_{\varphi_{*}}) J_{\varphi_{*}} \\
&\ge J_{\varphi_{*}} H_{\varphi_{*}}P_{\varphi_{*}}J_{\varphi_{*}} + \lambda^{} e_n J_{\varphi_{*}}^2\end{split}
\end{equation*}
and also 
\begin{equation*} 
\begin{split}
J_{\varphi_{\pm} } H_{\varphi_{\pm} } J_{\varphi_{\pm} }
&= J_{\varphi_{\pm} } H_{\varphi_{\pm} } P_{\varphi_{\pm} } J_{\varphi_{\pm} } + J_{\varphi_{\pm} } H_{\varphi_{\pm} } (\operatorname{id}-P_{\varphi_{\pm} }) J_{\varphi_{\pm} } \\
&\ge J_{\varphi_{\pm} } H_{\varphi_{\pm} }P_{\varphi_{\pm} }J_{\varphi_{\pm} } + \lambda^{2/3} f_m J_{\varphi_{\pm} }^2.
\end{split}
\end{equation*}
In particular, we find 
\[ \Vert  V_1' \Vert^2_{\mathbb R^n} \ge c\lambda^{-6/5} \text{ on } J\text{ for some }c>0.\] 
and 
\[ \Vert V_1'' \Vert_{\mathbb R^n} \ge c \lambda^{-3/10}  \text{ on } J\text{ for some }c>0.\] 
This implies for large $\lambda$ that
\begin{equation}
\label{eq:support2}
 J H J \ge \lambda^{2/3} f_m J^2.
 \end{equation}
From \eqref{eq:usefulrepresent3} we then conclude that for some $C>0$
\[ \mathcal H \ge \lambda^{2/3} f_m J^2 - C \lambda^{4/9} +\sum_{i \in \left\{\varphi_{\pm},\varphi_{*} \right\}} J_i H_i P_i J_i = \lambda^{2/3} f_m+ \sum_{i \in \left\{\varphi_{\pm},\varphi_{*} \right\}} J_i H_i P_i J_i -o(\sqrt{\lambda}). \]
This implies the claim of the Proposition, since \[\operatorname{rank} \left(J_{0} H_{\pm}PJ_{0} \right) \le n.\] More precisely, for the eigenvalues $E_1(\lambda)\le E_2(\lambda)\le..$ of $\mathcal H$ we have shown that 
\[\liminf_{\lambda \rightarrow \infty} \lambda^{-2/3}E_n(\lambda)\ge f_{n-1}>0 \text{ for n} \ge 2 .\] 
In particular, the lowest possible eigenvalue $e_1=0$ of the renormalized Schr\"odinger operator is of course attained as the nullspace of the renormalized Schr\"odinger operator $H_{\varphi_{*}}$ is non-trivial.
This shows that the spectral gap of the renormalized Schr\"odinger operator grows at least proportional to $\lambda^{2/3}.$ 
\end{proof}

\section{Multi-component $\mathcal O(n)$-models}
\label{sec:High-dim}
\begin{figure}
\centerline{\includegraphics[height=7cm]{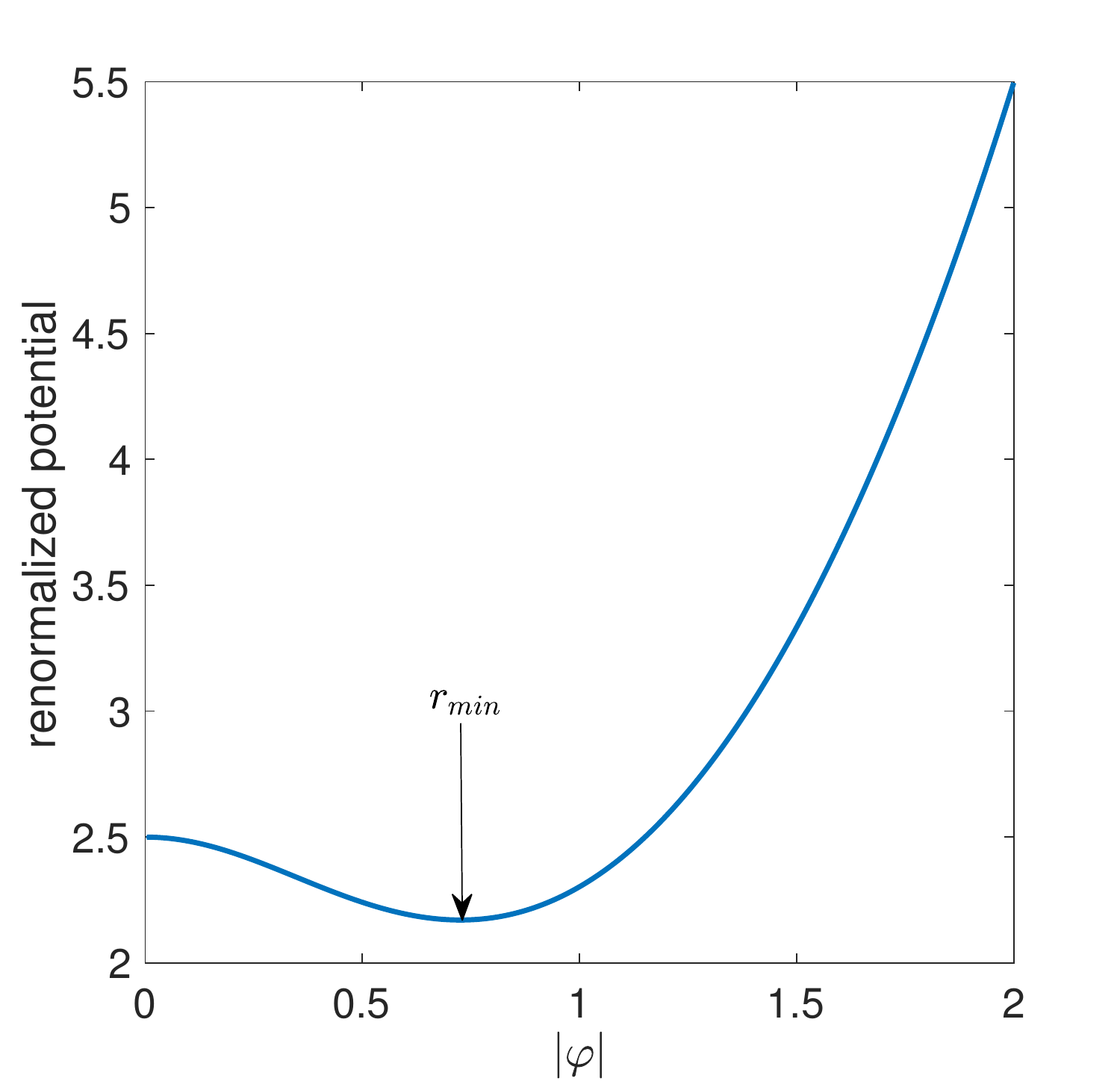}} 
\caption{\emph{Heisenberg model,} ($n=3$): The renormalized potential of the Heisenberg model for $h=0$ and $\beta=5$.}
\label{Figure3}
\end{figure}
\subsection{$n\ge 2$: Zero magnetic field, $h=0$}
\label{sec:ZMF and MC}
Let $h=0$ then the renormalized potential for $n \ge 2$ is radially symmetric and possesses a critical point at $\varphi=0$. In the supercritical case, i.e.\@ $\beta>n,$ the renormalized potential possesses another critical radius $r=\Vert \varphi \Vert \in (0,1),$ see Figure \ref{Figure4}. 
To see this, we differentiate the renormalized potential
\[ \partial_r V_n(r) = \beta r \left(1- \tfrac{I_{n/2}(\beta r)}{r I_{n/2-1}(\beta r)} \right). \]
It is now obvious that $r=0$ is a critical point of the renormalized potential at which 
\begin{equation}
\begin{split}
&\lim_{r \downarrow 0}\tfrac{I_{n/2}(\beta r)}{r I_{n/2-1}(\beta r)}= \tfrac{2}{n} \tfrac{\beta}{2}>1 \text{ such that } \partial_r^2 V_n(0)=\beta\left(1-\tfrac{\beta}{n} \right)<0
\end{split}
\end{equation}
where we used that $\beta>n$ is supercritical.
To conclude the existence of precisely one other critical radius $r_{\text{min}}$ at which the renormalized potential attains its global minimum it suffices therefore to show that $\tfrac{I_{n/2}(\beta r)}{r I_{n/2-1}(\beta r)}$ decays monotonically to zero. We prove this in Lemma \ref{eq:auxlemma} in the appendix. This implies that also the factor $\left(1- \tfrac{I_{n/2}(\beta r)}{r I_{n/2-1}(\beta r)} \right)$ has precisely one root, i.e. the second critical radius.

In the next proposition we show that the radial part of the measure $d\nu_N(\varphi)$ which we denote by $dr_N:=\nu_N^{(n)} r^{n-1} e^{-NV_n(r)} \ dr$ in the sequel satisfies a {SGI} with a constant that is uniformly bounded in the number of spins.

\subsection{Zero magnetic field- A lower bound on the spectral gap}
\begin{figure}
\centerline{\includegraphics[height=3.5cm]{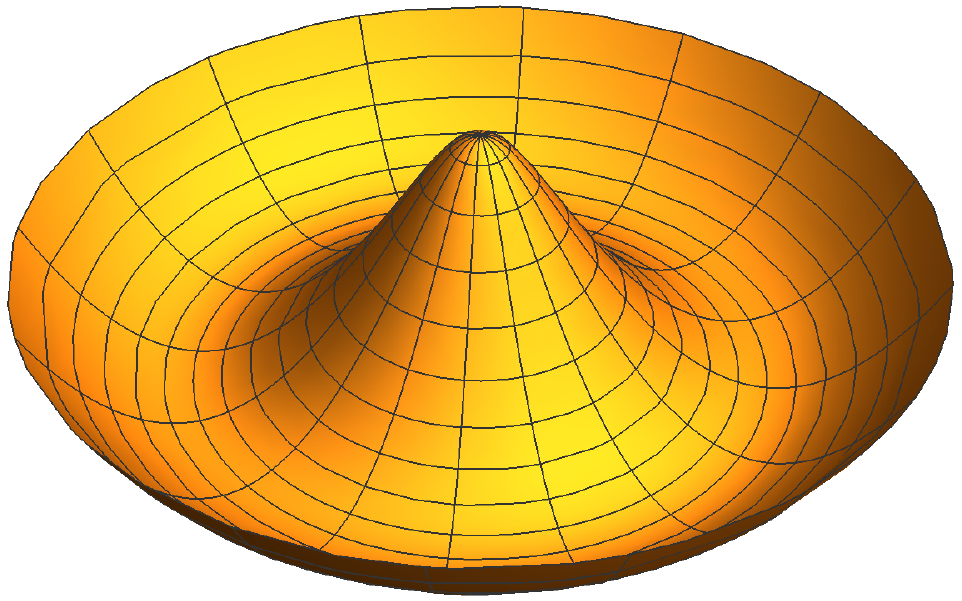}} 
\caption{\emph{XY-model:} The renormalized potential of the XY-model for $h=0$ and $\beta=10$.}
\label{Figure4}
\end{figure}

When $h=0$ and $n \ge 2$, then the renormalized Schr\"odinger operator \eqref{eq:renWitt} for $\lambda:=N/2$ is the self-adjoint operator \[\Delta_{\text{ren}}= - \Delta_{\mathbb R^n} + \lambda^2 \left\lvert \nabla_{\mathbb R^n}  V \right\rvert^2 - \lambda \Delta_{\mathbb R^n} V.\]
This operator is also rotationally symmetric such that by separating (spherical coordinates) the angular part from the radial part, the remaining radial component $\Delta^{\text{rad},\ell}_{\text{ren}}$ of the renormalized Schr\"odinger operator on $L^2((0,\infty),r^{n-1} dr )$  for $\ell \in \mathbb N_0$ reads
\begin{equation}
\label{eq:radialone}
\Delta^{\text{rad},\ell}_{\text{ren}} = - \left(\partial_r^2  + \tfrac{n-1}{r} \partial_r - \tfrac{\ell(\ell+n-2)}{r^2}\right) +  \lambda^2 \left\lvert \partial_r V_n(r) \right\rvert^2-\lambda \partial_r^2 V_n(r).
\end{equation}
Here, the term $\ell(\ell+n-2)$ accounts for the eigenvalues of the angular part of the Laplacian.
The renormalized potential possesses, when $h=0$ and $n \ge 2$, exactly two critical radii at which $\left\lvert \partial_r V_n(r) \right\rvert^2=0.$ The radii are $r=0$ and $r=r_{\text{min}},$ see the beginning of this Section \ref{sec:ZMF and MC}.
However, $V_n(r)$ is strictly concave at $0$, i.e. $ \partial_r^2 V_n(0)<0,$ and by Lemma \ref{eq:auxlemma} strictly convex at $r_{\text{min}}$ such that $ \partial_r^2 V_n(r_{\text{min}})>0$. This follows from
\[\partial_r^2 V_n(r_{\text{min}}) =\beta r_{\text{min}} \partial_r \vert_{r=r_{\text{min}}}\left(1- \tfrac{I_{n/2}(\beta r)}{r I_{n/2-1}(\beta r)} \right)>0,\]
see the beginning of Section \ref{sec:High-dim}.

By the tensorization principle we already know that the rotational invariance of the renormalized measure implies that the spectral gap inequality for the renormalized measure is at least uniform in $N.$ In our next Proposition we therefore study the low-lying spectrum of the radial component, $\Delta_{\text{ren}}^{\text{rad},0}$, as $\lambda \rightarrow \infty$. 
\begin{prop}
\label{prop:radialop}
Let $h=0$, $n \ge 2$ and $\beta >n.$ The spectral gap of the radial renormalized Schr\"odinger operator $\Delta_{\text{ren}}^{\text{rad},0}$ grows linearly in $N$.
\end{prop}
\begin{proof}
To study the low-lying spectrum of the radial component of the renormalized Schr\"odinger operator, let $\lambda:= N/2$ and consider Schr\"odinger operators
\begin{equation}
\begin{split}
\label{eq:Hamiltonian}
&H^{0}_{\text{osc}}(\lambda)=  - \left(\partial_r^2  + \tfrac{n-1}{r} \partial_r\right)+ \lambda^2 \vert\partial_r^2 V_n(0) \vert^2 r^2 - \lambda \ \partial_r^2 V_n(0) \text{ and} \\
&H^{r_{\text{min}}}_{\text{osc}}(\lambda)= - \partial_x^2 + \lambda^2 \ \vert \partial_r^2 V_n(r_{\text{min}}) \vert^2(x-r_{\text{min}})^2-\lambda \ \partial_r^2 V_n(r_{\text{min}})
\end{split}
\end{equation}
where we use the variable $x$ rather than $r$ to emphasize that the last operator is defined on $L^2(\mathbb R)$, unlike the first one which is an operator on $L^2((0,\infty),r^{n-1} \ dr).$ Observe that in \eqref{eq:Hamiltonian} we replaced the gradient term of the Schr\"odinger operator by its Taylor approximation at the critical point. This explains the occurrence of the second derivative at the critical point in \eqref{eq:Hamiltonian}.
Invoking the unitary maps $U_0 \in \mathcal L( L^2((0,\infty),r^{n-1} \ dr))$ and $U_{r_{\text{min}}}  \in \mathcal L(L^2(\mathbb R))$ defined as 
\begin{equation}
\begin{split}
(U_{0}f)(x) = \lambda^{-n/4}f(\lambda^{-1/2}x)  \text{ and } (U_{r_{\text{min}}}f)(x) = \lambda^{-1/4}f(\lambda^{-1/2}(x+r_{\text{min}}))
\end{split}
\end{equation}
shows that the two Schr\"odinger operators in \eqref{eq:Hamiltonian} are in fact unitarily equivalent, up to multiplication by $\lambda$, to the $\lambda$-independent Schr\"odinger operators
\begin{equation}
\begin{split}
S^{0}_{\text{osc}} &= - \left(\partial_r^2  + \tfrac{n-1}{r} \partial_r\right)+ \vert\partial_r^2 V_n(0) \vert^2 r^2 -  \ \partial_r^2 V_n(0)\\
S^{r_{\text{min}}}_{\text{osc}}&= - \partial_x^2 +  \ \vert \partial_r^2 V_n(r_{\text{min}}) \vert^2x^2-\ \partial_r^2 V_n(r_{\text{min}}),
\end{split}
\end{equation}
respectively. More precisely, we have that
\begin{equation}
\begin{split}
\label{eq:uniteqv}
U_{0}^{-1}\lambda \ S^{0}_{\text{osc}} U_{0}&= H^{0}_{\text{osc}}(\lambda)\quad\text{ and }\quad
U_{r_{\text{min}}}^{-1} \lambda \ S^{r_{\text{min}}}_{\text{osc}}U_{r_{\text{min}}}= H^{r_{\text{min}}}_{\text{osc}}(\lambda).
\end{split}
\end{equation}
Since the bottom of the spectrum of the operator $S^0_{\text{osc}}$ is strictly positive $S^0_{\text{osc}}\ge -\partial_r^2 V_n(0)>0,$
we conclude from \eqref{eq:uniteqv} that the bottom of the spectrum of $H^{0}_{\text{osc}}(\lambda)$ increases linearly to infinity as $\lambda \rightarrow \infty.$

To connect the low-energy spectrum of the renormalized Schr\"odinger operator with the above auxiliary operators, take $j \in C_c^{\infty}(-\infty,2)$ such that $j(x)=1$ for $\left\lvert x \right\rvert \le 1.$ Then, we define
\begin{equation}
\begin{split}
\label{eq:property}
J_0(x) &= j(\lambda^{2/5}\left\lvert x \right\rvert ), \ J_{r_{\text{min}}}(x) = j(\lambda^{2/5}\left\lvert x-r_{\text{min}} \right\rvert )\text{ with }\left\lVert \nabla  J_{i} \right\rVert_{\mathbb R^n} = \mathcal O(\lambda^{2/5})
\end{split}
\end{equation}
for $i \in \left\{ 0,r_{\text{min}} \right\}$ and $J := \sqrt{1-J_{r_{\text{min}}}^2-J_0^2}.$ 
\medskip

Without loss of generality we can assume that $\lambda$ is large enough such that $J_0$ and $J_{r_{\text{min}}}$ are disjoint.

\medskip

Taylor expansion of the potential at $0$ and $r_{\text{min}}$ respectively and the estimate on the gradient \eqref{eq:property} imply that  
\[\left\lvert J_i(\Delta^{\text{rad},0}_{\text{ren}}-H^{i}_{\text{osc}})J_i \right\rvert = \mathcal O(\lambda^{4/5}) \text{ for }i \in \left\{ 0,r_{\text{min}} \right\}. \]
Let $0=e_1<e_2\le..$ be the eigenvalues (counting multiplicities) of $S^{0}_{\text{osc}} \oplus S^{r}_{\text{osc}}$ and choose $\tau$ such that $e_{n+1}>\tau >e_n$ with $P_i$ being the projection onto the eigenspace to all eigenvalues of $H^{i}_{\text{osc}}$ below $\tau \lambda.$
The IMS (\textbf{I}smagilov, \textbf{M}organ, and \textbf{S}imon/\textbf{S}igal) formula, see \cite[(11.37)]{CFKS} for a version on manifolds, implies that  
\[ \Delta^{\text{rad},0}_{\text{ren}} = J \Delta^{\text{rad},0}_{\text{ren}} J  - \vert \partial_r J \vert^2 +\sum_{i \in \left\{0,r_{\text{min}} \right\}} \left(J_{i} \Delta^{\text{rad},0}_{\text{ren}} J_{i}-  \vert \partial_r J_{i} \vert^2 \right) \]
such that 
\begin{equation}
\label{eq:usefulrepresent}
\Delta^{\text{rad},0}_{\text{ren}} = J \Delta^{\text{rad},0}_{\text{ren}} J -\vert \partial_r J \vert^2+\sum_{i \in \left\{0,r_{\text{min}} \right\}} \left(J_{i} H^{i}_{\text{osc}} J_{i} + J_{i} (\Delta^{\text{rad},0}_{\text{ren}}-H^{i}_{\text{osc}}) J_{i}-  \vert \partial_r J_{i} \vert^2 \right).
\end{equation}
On the other hand, it follows that 
\begin{equation*} 
\begin{split}
J_{i} H^{i}_{\text{osc}} J_{i}
&= J_{i} H^{i}_{\text{osc}} P_i J_{i} + J_{i} H^{i}_{\text{osc}} (\operatorname{id}-P_i) J_{i} \ge J_{i} H^{i}_{\text{osc}}P_i J_{i} + \lambda e_n J_{i}^2. 
\end{split}
\end{equation*}
By construction, since $\nabla V_n$ vanishes linearly on the support of $J_i,$ we have 

\[ \Vert \nabla V_n \Vert^2_{\mathbb R^n} \ge c (\lambda^{-2/5})^2 =c\lambda^{-4/5} \text{ on } J\text{ for some }c>0.\] 
Since $\Delta V_n$ is globally bounded anyway, this implies for large $\lambda$ that
\begin{equation}
\label{eq:support}
 J \Delta^{\text{rad},0}_{\text{ren}} J \ge J^2 (c \lambda^{6/5}-\lambda) \ge \lambda e_n J^2.
 \end{equation}
From \eqref{eq:usefulrepresent} we then conclude that for some $C>0$
\[ \Delta^{\text{rad},0}_{\text{ren}} \ge \lambda e_n - C \lambda^{4/5} + \sum_{i \in \left\{0,r_{\text{min}} \right\}}J_{i} H^{i}_{\text{osc}}P_i J_{i}= \lambda e_n+ \sum_{i \in \left\{0,r_{\text{min}} \right\}}J_{i} H^{i}_{\text{osc}}P_i J_{i}-o(\lambda). \]
This implies the claim of the Proposition, since \[\operatorname{rank} \left(\sum_{i \in \left\{0,r_{\text{min}} \right\}}J_{i} H^{i}_{\text{osc}}P_i J_{i} \right) \le n.\] More precisely, for the eigenvalues $E_1(\lambda)\le E_2(\lambda)\le..$ of $\Delta^{\text{rad},0}_{\text{ren}}$ we have shown that 
\[\liminf_{\lambda \rightarrow \infty} \lambda^{-1}E_n(\lambda)\ge e_n.\] 
In particular, the lowest possible eigenvalue $e_1=0$ of the renormalized Schr\"odinger operator is of course attained as the nullspace of the renormalized Schr\"odinger operator is non-trivial.
This shows that the spectral gap of the renormalized Schr\"odinger operator grows at least linearly in $\lambda$ in the angular sector $\ell=0$. 
\end{proof}
\begin{corr} Let $h=0$, $n \ge 2$ and $\beta >n$. The spectral gap of the full Gibbs measure $\rho$ does not close faster than $\Theta(1/N).$ In particular, for radial functions, i.e. $f$ only depends on $\vert \bar \sigma \vert,$ the spectral gap remains open.
\end{corr}
\begin{proof}
Since the spectral gap of the radial component of the renormalized measure grows linearly in $N$ and the spectral gap of the angular component is uniform in $N$, the tensorization principle implies that the full renormalized measure satisfies a SGI that is uniform in $N$. 
Due to Proposition \ref{theo1}, the spectral gap of the full measure does therefore not close faster than of order $1/N.$

For radial functions $f$, the  $\RR^n \ni \varphi \mapsto E_{\mu_{\varphi}}(f)$ maps also into radial functions and therefore the spectral gap of the renormalized measure is only determined by the radial renormalized Schr\"odinger operator in Prop. \ref{prop:radialop}. Using Proposition \ref{theo1} and \eqref{eq:SGIstep}, this implies that for radial functions, the gap remains open.
\end{proof}

In the next Proposition we show that the rate $N^{-1}$ in this case is in fact optimal: 
\begin{prop} Let $h=0$, $n \ge 2$ and $\beta >n$. The spectral gap of the full measure $\rho$ of the dynamics decays at least as fast as $N^{-1}$. 
\end{prop}
\begin{proof}
We consider the mean-spin $\bar \sigma:(\mathbb S^n)^N \rightarrow \RR^{n+1}$ defined by

\begin{align} 
\bar \sigma(\sigma) = \frac{1}{N} \sum_{i=1}^N \sigma_i  = (\bar \sigma_1(\sigma),...,\bar \sigma_{n+1}(\sigma)) \in \RR^{n+1}.
\end{align} 

In analogy to the spherical harmonics which in cartesian coordinates reads $x_1/\Vert x \Vert$, we consider the function:  
\begin{align} 
f(\sigma):= \frac{\bar \sigma_1(\sigma)}{\Vert \bar \sigma(\sigma) \Vert} \eta( \Vert \bar \sigma(\sigma) \Vert).
\end{align} 
where $\eta \in C^{\infty}(\mathbb{R};[0,1])$ is a cut-off function such that for fixed $1>\delta>0$: \begin{align} 
\eta(t):=\left\{ \begin{array}{ll} 1, &\text{when}\ t> \delta\ \text{and}\\
0, &\text{when}\  t  \le \delta/2  .
 \end{array} \right. \end{align}
 As we want to compute the covariant derivative $\nabla_{\sigma_1} t(\sigma), $ we consider the parametrisation 
 $\gamma_1(t)$ so that $\gamma_1(0)= \sigma_1$. Then we define $\gamma(t):=(\gamma_1(t), \sigma_2, \dots, \sigma_N)$ and $s(t):= \bar \sigma(\gamma(t))$. It is then clear that for $v:=\gamma_1'(0)$ we have  $s'(0)=v/N$, the first coordinate of which is  $s_1'(0)=\langle e_1, v \rangle /N$. We define then 
\begin{align} 
s_1(t):=\bar \sigma_1 (\gamma(t))=\langle e_1, s(t) \rangle.
\end{align} 

Thus, since $f(\gamma(t))=\frac{s_1(t)}{\Vert s(t) \Vert},$
we find for the derivative

\begin{equation}
\begin{split}
f'(0)&= \frac{1}{ \vert s(0) \vert^2} \left( \vert s(0) \vert s_1'(0)- s_1(0) \frac{s(0) \cdot s'(0)}{\vert s(0) \vert}  \right) \eta(\vert s(0)\vert)  \\
&\qquad + \frac{s_1(0)}{\vert s(0) \vert} \eta'(\vert s(0)\vert ) \frac{s(0) \cdot s'(0)}{\vert s(0) \vert}  \\
&= \frac{1}{N\vert \bar \sigma(\sigma) \vert^2} \left( \vert \bar \sigma(\sigma) \vert \langle e_1,v \rangle- \langle e_1, \bar \sigma(\sigma) \rangle \frac{\langle \bar \sigma(\sigma),v \rangle}{\vert \bar \sigma(\sigma) \vert}\right) \eta(\vert \bar \sigma(\sigma) \vert ) \\
&+  \frac{\bar \sigma_1(\sigma)}{\vert \bar \sigma(\sigma) \vert} \eta'(\vert \bar \sigma(\sigma ) \vert ) \frac{ \langle \bar \sigma(\sigma),v \rangle }{N \vert \bar \sigma(\sigma) \vert}   .
\end{split}
\end{equation}
Therefore, we see that in terms of $$ \mu(\sigma):= \frac{1}{N} \left( \frac{1}{\vert \bar \sigma(\sigma) \vert^2} \left(\vert \bar \sigma(\sigma) \vert e_1- \bar \sigma_1(\sigma) \frac{\bar \sigma(\sigma)}{\vert \bar \sigma(\sigma)\vert}\right)\eta(\vert \bar \sigma(\sigma) \vert ) +   
\frac{\bar \sigma_1(\sigma)}{\vert \bar \sigma(\sigma) \vert^2} \eta'(\vert \bar \sigma(\sigma ) \vert )   \bar \sigma(\sigma) 
 \right), $$
the derivative is just
$$\nabla_{\sigma_1} f(\sigma) = \mu(\sigma)-\langle \mu(\sigma), \sigma_1 \rangle \sigma_1.$$
The cut-off function $\eta$ ensures that $\vert \bar \sigma (\sigma) \vert $ is not small. Therefore we can bound $| Z(\sigma)| =\mathcal O(1)$ which implies that $ \mathbb{E}_{\rho}(| \nabla_{\sigma_1}\bar \sigma(\sigma)|^2) \lesssim 1/N^2$ or that $$ \sum_{i \in [N]} \mathbb{E}_{\rho}(| \nabla_{\sigma_i}\bar \sigma(\sigma)|^2) \lesssim 1/N.$$
 
By rotational symmetry we also know that $\mathbb E_{\rho}(f)=0$. For the second moment $\mathbb E_{\rho}(f(\sigma)^2)$, we have by rotational invariance again
\begin{align} 
1&=\sum_{i=1}^{n+1}\mathbb E_{\rho}\left( \frac{\bar \sigma_i(\sigma)^2}{\vert \bar \sigma(\sigma) \vert^2} \right)= \sum_{i=1}^{n+1}\mathbb E_{\rho}\left( \frac{\bar \sigma_1(\sigma)^2}{\vert \bar \sigma(\sigma) \vert^2} \right)\\ &= (n+1)\mathbb E_{\rho}\left( \frac{\bar \sigma_1(\sigma)^2}{\vert \bar \sigma(\sigma) \vert^2} \eta\left(\vert \bar \sigma(\sigma)  \vert\right) \right) + (n+1)\mathcal{R}_N \notag
\end{align} 
where $\mathcal{R}_N$ is the error 
\[\mathcal{R}_N:= \mathbb E_{\rho}\left( \frac{\bar \sigma_1(\sigma)^2}{\vert \bar \sigma(\sigma) \vert^2} \left(1- \eta \left(  \vert \bar \sigma(\sigma)  \vert  \right) \right) \right).\]

Our aim is now to argue that $ \mathcal{R}_N$ is small as $N$ is large.

For $\beta>n$ we know that the renormalized potential attains its minimum at hyperspheres $\partial B_{\RR^n}(0,r_{\text{min}}).$ This implies that the renormalized measure concentrates at such $\varphi \in \partial B(0,r_{\text{min}})$ with exponential tail bounds, i.e. the probability of $\varphi$ away from $\partial B(0,r_{\text{min}})$ is exponentially small in $N$. 
The fluctuation measure then enforces that also the mean spin $\bar \sigma $ has to be outside of a ball of radius $\delta>0$ with high probability.
To see this recall that the fluctuation measure can be rewritten as 
\begin{equation}
 \mathbb E_{\mu_{\varphi}}(F) = \frac{\int_{\left(\mathbb S^{n-1}\right)^N} F(\sigma)  e^{\beta N \langle \varphi, \bar \sigma \rangle} \ dS_{\mathbb S^{n-1}}^{\otimes^N}(\sigma)}{\int_{\left(\mathbb S^{n-1}\right)^N} e^{\beta N \langle \varphi, \bar \sigma \rangle} \ dS_{\mathbb S^{n-1}}^{\otimes^N}(\sigma)} =\frac{\int_{\left(\mathbb S^{n-1}\right)^N} F(\sigma)  e^{\beta N \langle \varphi, \bar \sigma \rangle} \ dS_{\mathbb S^{n-1}}^{\otimes^N}(\sigma)}{\mathcal N(\varphi)^N} .
 \end{equation}
 Here, the radial normalizing function 
 \[\mathcal N(\varphi):=\int_{\mathbb S^{n-1}} e^{\beta \langle \varphi,  \sigma \rangle} \ dS_{\mathbb S^{n-1}}(\sigma)=\Gamma\left(\tfrac{n}{2}\right) \left(\tfrac{2}{ \left\lVert \beta \varphi  \right\rVert}\right)^{\frac{n}{2}-1} I_{\tfrac{n}{2}-1}( \left\lVert \beta \varphi  \right\rVert) \] is a strictly monotonically increasing function of $\vert \varphi\vert$ that satisfies $\mathcal N(\varphi) \ge 1$ and $\mathcal N(\varphi)=1$ if and only if $\varphi=0$. This follows directly from the Taylor series of the modified Bessel function.
 
Hence, we can pick $\delta$ such that $e^{\beta \langle \varphi, \bar \sigma \rangle} < (1+\mathcal N(\varphi))/2$ for all $\varphi \in \partial B(0,r_{\text{min}})$ and $\vert \bar \sigma \vert \le \delta.$ Hence, we see that for such $\varphi$ 
\[ \mathbb E_{\mu_{\varphi}}(\indic_{\vert \bar \sigma \vert \le \delta}) = \mathcal O\left(\left( \frac{1}{2\mathcal N(\varphi)}+\frac{1}{2} \right)^N \right).\]
This shows that $ \mathbb E_{\rho}(\indic_{\vert \bar \sigma \vert \le \delta}) = \mathcal O\left(\left( \frac{1}{2\mathcal N(\varphi)}+\frac{1}{2} \right)^N \right)$ and hence that $\mathcal R_N$ tends exponentially fast to zero as well, as $N$ tends to infinity, under the condition that $\beta>n.$ 
\end{proof}

\subsection{Nonzero magnetic fields for $n\ge 2$}
\begin{figure}
\centerline{\includegraphics[height=3.0cm]{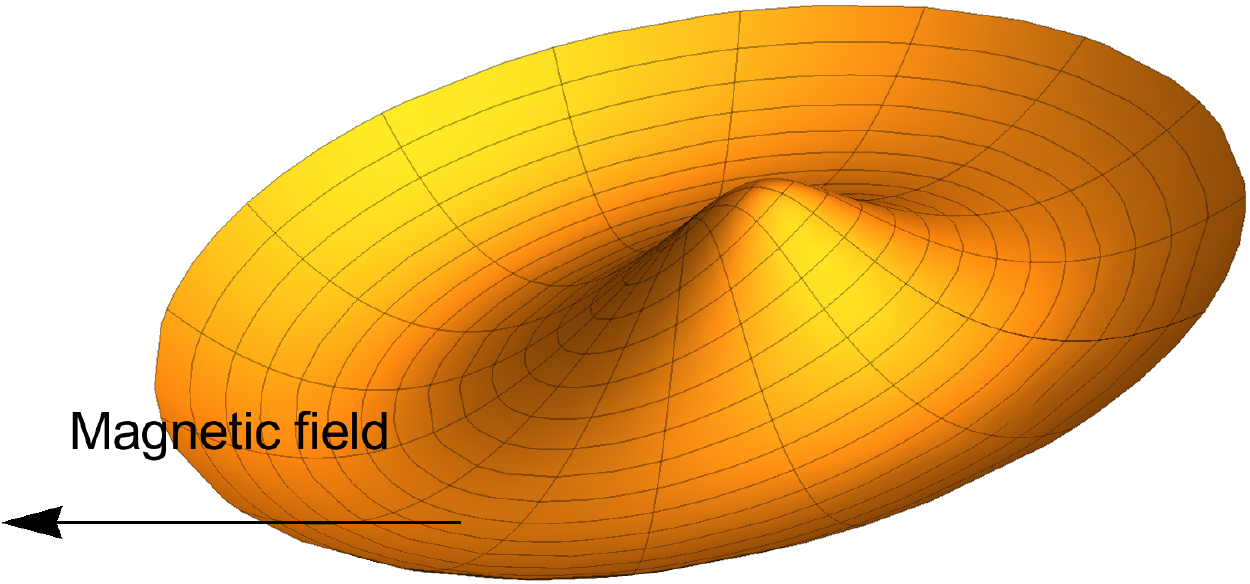}} 
\caption{\emph{XY-model:} The renormalized potential of the XY-model for $h=(-2,0)$ and $\beta=10$. The rotational symmetry is broken.}
\label{Figure5}
\end{figure}

The situation $h\neq 0$ and $n \ge 2$ cannot be reduced to a one-dimensional model due to lack of symmetries. Yet, the renormalized Schr\"odinger operator provides a very elegant tool to show that the spectral gap of the full generator of the Ginzburg-Landau dynamics dynamics remains open as $N \rightarrow \infty.$ 

\medskip

In fact, whereas the global minimum for $h=0$ of the renormalized potential is attained on a hypersphere, the global minimum for $h \neq 0$ is attained at a single point, only. This allows us to identify the asymptotic of the low-energy spectrum of the renormalized Schr\"odinger operator directly with the spectrum of a quantum harmonic oscillator.

\medskip

Let $\varphi_{\text{c}} \in \mathbb R^n$ be a critical point of the renormalized potential \eqref{eq:renpot}. We define the set 
\[ \Sigma:=\left\{ \sum_{i=1}^{n} \left(n_i\vert \lambda_i \vert + \tfrac{1}{2} \left(\vert \lambda_i \vert - \lambda_i\right)\right) , \text{ with } n_i \in \mathbb N_0, \lambda_i \in \sigma(D^2V_n(\varphi_{\text{c}})) \right\}\]
where $\lambda_1,...,\lambda_n$ comprise the entire spectrum of $D^2V_n(\varphi_{\text{c}})$.

\medskip

Let $e_k$ be the $k$-th smallest element counting multiplicity in $\Sigma$ we then have the following Proposition:
\begin{prop}
\label{Simon}
Let $h \neq 0$, $\beta \ge n$, and $n \ge 2$.  Let $E_k(\lambda)$ denote the $k$-th lowest eigenvalue of the renormalized generator then this eigenvalue satisfies the asymptotic law $\lim_{\lambda \rightarrow \infty} \frac{E_k(\lambda)}{\lambda} = e_{k}.$
In particular, the ground state of the renormalized generator in the limit as $\lambda \rightarrow \infty$ is unique and the spectral gap of the renormalized Schr\"odinger operator remains open and linearly in $\lambda$.
\end{prop}
\begin{proof}
When $h \neq 0$ then the renormalized potential has a unique non-degenerate minimum. To see this recall that the renormalized potential reads
\begin{equation*}
\begin{split}
V_n(\varphi,h)
& = \tfrac{\beta}{2}\left(1+\left\lVert \varphi \right\rVert^2 \right) - \log \left( \Gamma\left(\frac{n}{2}\right) \left(\tfrac{2}{ \left\lVert\beta \varphi +h\right\rVert}\right)^{\frac{n}{2}-1} I_{\tfrac{n}{2}-1}(\left\lVert \beta \varphi +h\right\rVert) \right).
\end{split}
\end{equation*}
Introducing the new variable $\zeta:=\beta \varphi + h$ implies that 
\begin{equation}
\begin{split}
\label{eq:potentialreni}
V_n(\varphi(\zeta),h)
& = \tfrac{1}{2\beta}\left(\beta^2+\left\lVert \zeta-h \right\rVert^2 \right) - \log \left( \Gamma\left(\frac{n}{2}\right) \left(\tfrac{2}{ \left\lVert\zeta \right\rVert}\right)^{\frac{n}{2}-1} I_{\tfrac{n}{2}-1}(\left\lVert \zeta\right\rVert) \right) \\
&= \tfrac{1}{2\beta }\left(\beta^2+\left\lVert \zeta \right\rVert^2+\left\lVert h\right\rVert^2 -2  \langle \zeta,h \rangle \right) - \log \left( \Gamma\left(\frac{n}{2}\right) \left(\tfrac{2}{ \left\lVert\zeta \right\rVert}\right)^{\frac{n}{2}-1} I_{\tfrac{n}{2}-1}(\left\lVert \zeta\right\rVert) \right).
\end{split}
\end{equation}
Computing the gradient of that expression yields 
\begin{equation}
\begin{split}
\label{eq:gradiennt}
\nabla_{\zeta} V_n(\varphi(\zeta),h)= -\frac{1}{\beta}h + g_{\beta}(\Vert \zeta \Vert) \widehat{e}_{\zeta}
\end{split}
\end{equation}
where we introduced the auxiliary function $g_{\beta}(r):= \left(\frac{r}{\beta} -  \frac{I_{n/2} (r)}{I_{n/2-1}( r)}\right)$. Thus for the gradient to vanish the vectors $h$ and $\zeta$ have to be linearly dependent. 

Assuming thus that $\widehat{e}_h = \pm \widehat{e}_{\zeta}$ we obtain from setting the gradient to zero the following equation 
\[\frac{ \beta I_{n/2} (\left\lVert \zeta \right\rVert) }{ I_{n/2-1}( \left\lVert \zeta \right\rVert)}=( \left\lVert \zeta \right\rVert\mp \left\lVert h \right\rVert).\]
Thus, when $h$ and $\zeta$ are aligned, there is precisely one solution, the global minimum of the renormalized potential, satisfying 
\[\frac{ \beta I_{n/2} (\left\lVert \zeta \right\rVert) }{ I_{n/2-1}( \left\lVert \zeta \right\rVert)}=( \left\lVert \zeta \right\rVert- \left\lVert h \right\rVert)\]
with $g_{\beta}(\Vert \zeta \Vert)= \beta^{-1}\Vert h \Vert>0.$ That the aligned scenario corresponds to the global minimum is evident from the expression of the renormalized potential \eqref{eq:potentialreni}.

The simplicity of the solution follows since the left hand side $\tfrac{ \beta I_{n/2} (\left\lVert \zeta \right\rVert) }{ I_{n/2-1}( \left\lVert \zeta \right\rVert)}$ is a concave, monotonically increasing function from $0$ to $\beta$ as $\left\lVert \zeta \right\rVert \rightarrow \infty.$

When $h$ and $\zeta$ point in opposite directions, there can, by concavity of the left-hand side, be between zero and two solutions to the equation
\[\frac{ \beta I_{n/2} (\left\lVert \zeta \right\rVert) }{ I_{n/2-1}( \left\lVert \zeta \right\rVert)}=( \left\lVert \zeta \right\rVert+ \left\lVert h \right\rVert)\]
with $g_{\beta}(\Vert \zeta \Vert)= -\beta^{-1}\Vert h \Vert<0.$
In particular, for sufficiently low temperatures there exists a local maximum and a saddle point of the renormalized potential as shown in Figure \ref{Figure5}.

From differentiating \eqref{eq:gradiennt}, the Hessian is given by
\begin{equation}
\begin{split}
\label{eq:Hessian}
D^2_{\zeta} V_n(\varphi(\zeta),h)= g'(\Vert \zeta \Vert)\frac{\zeta \zeta^T}{\Vert \zeta \Vert^2}+g_{\beta}(\Vert \zeta \Vert)\left(\frac{\operatorname{id}}{\Vert \zeta \Vert}-\frac{\zeta \zeta^T}{\Vert \zeta \Vert^3}\right).
\end{split}
\end{equation}
We note that the Hessian has full rank unless at critical points unless $g_{\beta}'(\Vert \zeta \Vert)+g_{\beta}(\Vert \zeta \Vert)(\Vert \zeta \Vert - \Vert \zeta \Vert^{-1})=0,$ since $g_{\beta}(\Vert \zeta \Vert)\neq 0$ by \eqref{eq:gradiennt} for non-zero magnetic fields.

\medskip

In addition, there can be only a saddle point which can only happen at one fixed temperature depending on $n$. 

\medskip

Finally, if the temperature is sufficiently high, yet still such that $\beta>n$, there may be no critical point if $h$ and $\zeta$ point in opposite directions. This is in particular the case when $\beta=n$ and $h \neq 0$: Taylor expansion at zero yields
\[ \frac{ \beta I_{n/2} (\left\lVert \zeta \right\rVert) }{ I_{n/2-1}( \left\lVert \zeta \right\rVert)} = \frac{\beta \Gamma(n/2)}{2\Gamma(1+n/2)} \Vert \zeta \Vert + \mathcal O(\Vert \zeta  \Vert^2)\] 
where for $\beta=n$ we find $\frac{n \Gamma(n/2)}{2\Gamma(1+n/2)}=1$ and concavity of the function $\Vert \zeta \Vert \mapsto \frac{ \beta I_{n/2} (\left\lVert \zeta \right\rVert) }{ I_{n/2-1}( \left\lVert \zeta \right\rVert)} $ show.

Thus $\Vert \nabla V_n \Vert^2$ vanishes at not more than three critical points $\varphi_{\text{c}}$ on the span of $h$. In particular, all eigenvalues of $D^2V_n$ are non-negative only at the global minimum of $V_n$ by \eqref{eq:Hessian}, since we already established that $g'(\Vert \zeta \Vert)<0$ at the other two. To see that they are strictly positive there, it suffices to analyze for $r=\Vert \zeta\Vert$
\begin{equation}
\begin{split}
g_{\beta}'(\Vert \zeta \Vert) &= \left(\beta^{-1} -\mathcal I(\Vert \zeta \Vert)^{-1} \right) + \Vert \zeta \Vert \frac{\mathcal I'(\Vert \zeta \Vert)}{\mathcal I(\Vert \zeta \Vert)^2} \\
&=\frac{g_{\beta}(\Vert \zeta \Vert)}{\Vert \zeta \Vert} + \frac{ \Vert \zeta \Vert \mathcal I'(\Vert \zeta \Vert)}{\mathcal I(\Vert \zeta \Vert)^2}>0.
\end{split}
\end{equation}

Hence, we find that
\begin{equation}
\begin{split}
g_{\beta}'(\Vert \zeta \Vert)+g_{\beta}(\Vert \zeta \Vert)(\Vert \zeta \Vert - \Vert \zeta \Vert^{-1}) = g_{\beta}(\Vert \zeta \Vert) \Vert \zeta \Vert +  \frac{ \Vert \zeta \Vert \mathcal I'(\Vert \zeta \Vert)}{\mathcal I(\Vert \zeta \Vert)^2}.
\end{split}
\end{equation}
In particular, this expression is strictly positive at the global minimum, since $g_{\beta}(\Vert \zeta \Vert) >0$ and $\mathcal I'(\Vert \zeta \Vert)>0$ by general principles, see Lemma \ref{eq:auxlemma}.

The asymptotic behaviour of the spectrum of the renormalized Schr\"odinger operator has been computed in \cite{S} and our above representation of $\Sigma$ follows by noticing that $\tfrac{1}{2} D^2  \left\lvert \nabla V_n \right\rvert^2(\varphi_{\text{c}}) = \left(D^2V_n(\varphi_{\text{c}})\right)^2>0.$ 

Since the renormalized Schr\"odinger operator and renormalized generator are unitarily equivalent up to a factor, the semiclassical eigenvalue distribution stated in \cite[Theorem $1.1$]{S} implies the statement of the Proposition. 
\end{proof}

\section{The critical regime, Proof of Theo. \ref{theo:theo2}}
\label{sec:critreg}
We conclude our analysis by investigating the  critical case $\beta=n$ and prove Theorem \ref{theo:theo2}. As before, we distinguish between $n=1$ and multi-component systems $n \ge 2:$
\subsection{Critical Ising model}
It follows from \eqref{eq:h=0}, which always vanishes for all $x>0$, that the spectral gap, at the critical point $\beta=n=1$, does not close exponentially fast in the number of spins. We want to show in this subsection that it closes at least polynomially, though. For a refined analysis in dimension $n=1$, we recall some basic ideas from discrete Fourier analysis: 
\begin{figure}
\centering
\begin{minipage}[b]{0.45\textwidth}
\includegraphics[width=\linewidth]{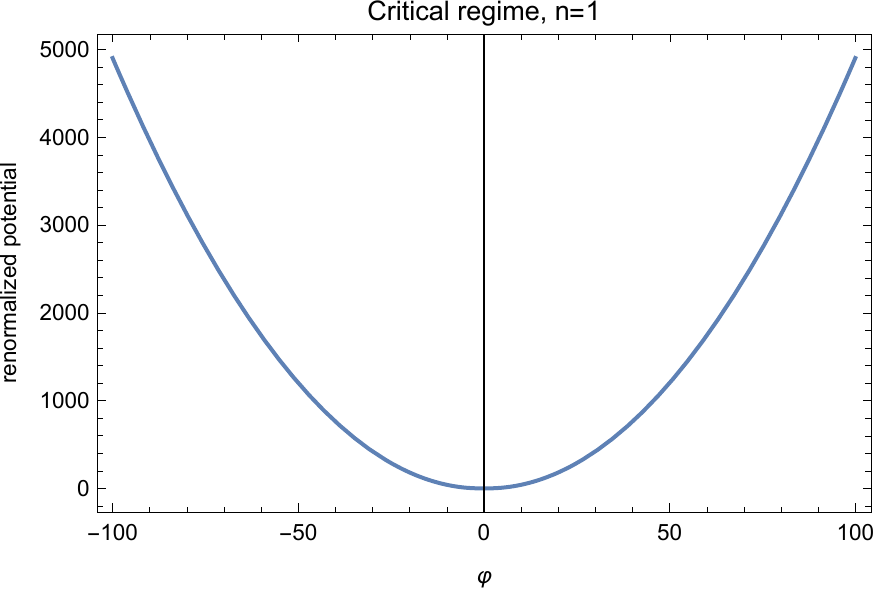}
\caption{The renormalized potential $V_1$ for $\beta=1$, $h=0$ is a symmetric convex function.}
\end{minipage}
\begin{minipage}[b]{0.45\textwidth}
\includegraphics[width=\linewidth]{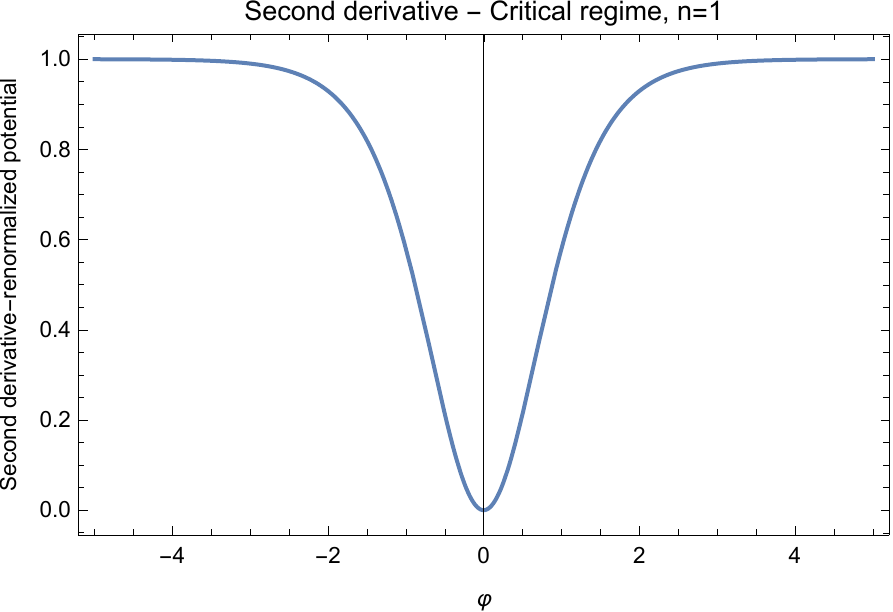} 
\caption{Strong convexity of the renormalized potential $V_1$ fails at the origin, $\varphi=0.$}
\end{minipage}
\label{Figurecrit}
\end{figure}

Let $f: \left\{ \pm 1 \right\}^N \rightarrow \mathbb C$ be an arbitrary function on the hypercube. The $L^2(\left\{ \pm 1\right\}^N,2^{-N}d\mu_{\text{count}})$ inner product on the hypercube is defined as 
\[\langle f,g \rangle_{\left\{ \pm 1\right\}^N} := \sum_{x \in \left\{ \pm 1\right\}^N} 2^{-N} f(x_1,..,x_N)\overline{g(x_1,..,x_N)}.\]
The characteristic function $\chi_S$ for $S \subset [N]$ is defined as $\chi_S(x):= \prod_{i \in S} \sigma_i$
and the family $(\chi_S)_{S \subset [N]}$ forms an orthonormal basis of $L^2(\left\{ \pm 1\right\}^N).$ In particular, $\chi_{\emptyset}=1.$
We also define indicator vectors $\indic_{S} \in \mathbb R^N$ such that $\indic_{S}(x)=1$ if $x \in S$ and $0$ otherwise. 

Every function $f \in L^2(\left\{ \pm 1\right\}^N)$ admits a unique Fourier decomposition

\begin{equation}
\label{eq:fourier}
 f = \sum_{ S \subset [N]} \widehat{f}(S) \chi_S
\end{equation}
where $\widehat{f}(S) := \langle f, \chi_S \rangle.$ The variance of the stationary measure is given as the sum of
\begin{equation}
\label{eq:variancedecomposition}
 \operatorname{Var}_{\rho}(f)=\mathbb E_{\nu_N}( \operatorname{Var}_{\mu_{\varphi}}(f))+ \operatorname{Var}_{\nu_{N}}(\mathbb E_{\mu_{\varphi}}(f)).
 \end{equation}
Since the first term on the right-hand side of this equation is always uniformly bounded by the Dirichlet form, as shown in the proof of Proposition \ref{theo1}, it suffices to study the behaviour of the second term.
Thus, applying the expectation with respect to the fluctuation measure yields by the Fourier decomposition \eqref{eq:fourier}
\begin{equation}\label{eq:fourier2}
 \mathbb E_{\mu_{\varphi}}(f)  = \sum_{ S  \subset [N]} \widehat{f}(S) \mathbb E_{\mu_{\varphi}}( \chi_S).
\end{equation}
In particular, using the explicit  form of $V_1(\phi)$, a direct computation yields for all $x \in [N]$
\[\mathbb E_{\mu_{\varphi}}( \sigma(x) )= e^{V_1(\varphi)} \frac{\left(e^{-\frac{\beta}{2} \vert \varphi-1 \vert^2+h }-e^{-\frac{\beta}{2}\vert \varphi+1 \vert^2-h } \right) }{2}=  \tanh(\beta \varphi+h). \]

Using that $\mu_{\varphi}$ is a product measure, this implies that the full expression for \eqref{eq:fourier2} is given by 
\begin{equation}
\begin{split}
\mathbb E_{\mu_{\varphi}}(f)&= \sum_{ S  \subset [N]} \widehat{f}(S)( \tanh(\beta \varphi+h))^{\vert S \vert}.
\end{split}
\end{equation}
Hence, we find for the variance 
\begin{equation}
\begin{split}
\label{eq:variance} \operatorname{Var}_{\nu_{N}}(\mathbb E_{\mu_{\varphi}}(f)) &=  \sum_{S_1,S_2 \subset[N]}^N\widehat{f}(S_1)\widehat{f}(S_2)\Big( \mathbb E_{\nu_N}\left( \tanh(\beta \varphi+h)^{\vert S_1\vert+\vert S_2\vert} \right)\\
& \qquad -\mathbb E_{\nu_N}\left( \tanh(\beta \varphi+h)^{\vert S_1\vert} \right)\mathbb E_{\nu_N}\left( \tanh(\beta \varphi+h)^{\vert S_2\vert} \right) \Big)
\end{split}
\end{equation}
For the Dirichlet form, we find, with $S_1 \triangle S_2$ denoting the symmetric difference of sets $S_1$ and $S_2$,
\begin{equation}
\label{eq:dirichlet}
\begin{split}
\sum_{x \in [N]} \mathbb E_{\rho}\left\lvert \nabla_{\mathbb S^{0}}^{(x)} f \right\rvert^2&= \sum_{x \in [N]}\sum_{ S_1,S_2  \subset [N]}\widehat{f}(S_1)\widehat{f}(S_2)\mathbb E_{\rho}(\nabla_{\mathbb S^{0}}^{(x)}\chi_{S_1} \nabla_{\mathbb S^{0}}^{(x)}\chi_{S_2}) \\
&= 4\sum_{x \in [N]}\sum_{ S_1,S_2  \subset [N]}\delta_{x \in S_1}\delta_{x \in S_2}\widehat{f}(S_1)\widehat{f}(S_2)\mathbb E_{\rho}(\chi_{S_1} \chi_{S_2}) \\
&= 4\sum_{ S_1,S_2  \subset [N]}\sum_{x \in S_1 \cap S_2}\widehat{f}(S_1)\widehat{f}(S_2)\mathbb E_{\rho}(\chi_{S_1\triangle S_2}) \\
&= 4\sum_{ S_1,S_2  \subset [N]} \langle \indic_{S_1}, \indic_{S_2} \rangle_{\mathbb R^N} \widehat{f}(S_1)\widehat{f}(S_2)\mathbb E_{\rho}(\chi_{S_1\triangle S_2})\\
&= \mathbb E_{\rho} \left\lvert 2\sum_{ S  \subset [N]} \indic_{S} \widehat{f}(S) \chi_S \right\rvert^2_{\mathbb R^N}.
\end{split}
\end{equation}
\begin{prop}
\label{prop1}
For zero magnetic fields, i.e. $h=0$, and $\beta \ge 1$,all functions with Fourier support on sets of fixed cardinality $k \in \mathbb N$, i.e. for $f$ given as 
\[f = \sum_{S \subset [N]; \left \lvert S \right\rvert = k} \widehat{f}(S) \chi_S.\]
satisfy the inequality $\operatorname{Var}_{\nu_{N}}(\mathbb E_{\mu_{\varphi}}(f)) \le \tfrac{N}{4k} \sum_{x \in [N]} \left\lvert  \nabla_{\mathbb S^{n-1}}^{(x)}f \right\rvert^2_{L^2(d\rho)}.$

In particular, for the magnetization
\begin{equation} 
\label{eq:mag}
M = \tfrac{1}{\sqrt{N}} \sum_{x \in [N]} \sigma(x) 
\end{equation}
we obtain an inequality 

\[ \operatorname{Var}_{\nu_{N}}(\mathbb E_{\mu_{\varphi}}(M)) =  \tfrac{N \mathbb E_{\nu_N} \left(\tanh( \beta \varphi)^{2}\right)}{4} \sum_{x \in [N]} \left\lvert  \nabla_{\mathbb S^{n-1}}^{(x)}M \right\rvert^2_{L^2(d\rho)}. \] 

Moreover, the spectral gap for critical $\beta=1$ closes at least like $\mathcal O(N^{-1/2}).$
\end{prop}
\begin{proof}
When $h=0$, it suffices to estimate the variance by Jensen's inequality as 
\begin{equation}
\begin{split}
\label{eq:bound1}
 \operatorname{Var}_{\nu_{N}}(\mathbb E_{\mu_{\varphi}}(f)) &\le \mathbb E_{\nu_N} \sum_{S_1,S_2 \subset [N] ; \left\lvert S_1 \right\rvert=\left\lvert S_2 \right\rvert=k}\widehat{f}(S_1)\widehat{f}(S_2)\mathbb E_{\mu_{\varphi}}(\chi_{S_1})\mathbb E_{\mu_{\varphi}}(\chi_{S_2})  \\
&\le \mathbb E_{\nu_N}\mathbb E_{\mu_{\varphi}} \left\lvert \sum_{S \subset[N] ; \left\lvert S_1 \right\rvert=\left\lvert S_2 \right\rvert=k}\widehat{f}(S) \chi_S \right\rvert^2  \\
&= \frac{1}{k^2} \mathbb E_{\rho} \left\lvert\left\langle \sum_{ S \subset[N] ; \left\lvert S_1 \right\rvert=\left\lvert S_2 \right\rvert=k} \widehat{f}(S)   \indic_{S} \chi_S,\indic_{[N]} \right\rangle_{\mathbb R^N} \right\rvert^2 \\
  &\le   \frac{N}{k^2}  \ \mathbb E_{\rho}\left\lvert \sum_{ S \subset [N] ;  \left\lvert S_1 \right\rvert=\left\lvert S_2 \right\rvert=k}\widehat{f}(S)\indic_{S}    \chi_S \right\rvert_{\mathbb R^N}^2.
\end{split}
\end{equation}
Using \eqref{eq:dirichlet} we then obtain the spectral gap inequality
\begin{equation}
\begin{split}
 \operatorname{Var}_{\nu_{N}}(\mathbb E_{\mu_{\varphi}}(f)) & \le   \frac{N}{k^2}  \ \mathbb E_{\rho}\left\lvert \sum_{ S \subset[N] ;\left\lvert S_1 \right\rvert=\left\lvert S_2 \right\rvert=k} \widehat{f}(S)  \indic_{S}   \chi_S \right\rvert_{\mathbb R^N}^2 \\
 & =  \frac{N}{4k^2}  \sum_{x \in [N]} \left\lvert \nabla_{\mathbb S^{0}}^{(x)} f \right\rvert_{L^2(d \rho)}^2.
  \end{split}
\end{equation}
Turning to the magnetization \eqref{eq:mag}, we can write the variance of the magnetization $M$ in terms of the expectation value $\mathbb E_{\nu_N} \left(\tanh( \beta \varphi)^{2}\right)$
\begin{equation}
\begin{split}
\label{eq:variance2}
\operatorname{Var}_{\nu_{N}}(\mathbb E_{\mu_{\varphi}}(M)) = \tfrac{1}{N} \sum_{ x,y \in [N]}  \mathbb E_{\nu_N}\left( \tanh( \beta \varphi)^{2}\right) = N \mathbb E_{\nu_N} \left(\tanh( \beta \varphi)^{2}\right).
\end{split}
\end{equation}
We now recall that $\tanh(\beta \varphi)^2 = \beta^2 \varphi^2 +\mathcal O(\varphi^4)$ and for $\beta =1$ 
\[ V_1(\varphi) = \frac{1}{2}+\frac{\varphi^4}{12}+\mathcal O(\varphi^6)\] by Taylor expanding around $0$.
It therefore follows from Laplace's principle \cite[Ch.\@ II,Theorem $1$]{W01} that 
\begin{equation}
\label{eq:Laplaceexp}
\mathbb E_{\nu_N} \left(\tanh( \beta \varphi)^{2}\right) \sim N^{1/4}N^{-3/4}=N^{-1/2}.
\end{equation}

On the other hand, we can compute the Dirichlet form of the magnetization using \eqref{eq:dirichlet}
\begin{equation}
\begin{split}
\label{eq:dirichlet2}
\sum_{x \in [N]} \mathbb E_{\rho}\left\lvert \nabla_{\mathbb S^{0}}^{(x)} M \right\rvert^2
&= \tfrac{4}{N} \mathbb E_{\rho} \left\lvert \sum_{ x  \in [N] } \indic_{S} \chi_S \right\rvert^2 =4  \mathbb E_{\rho}(1)=4.
\end{split}
\end{equation}
Thus, comparing \eqref{eq:variance2} with \eqref{eq:dirichlet2} implies the claim together with the asymptotic \eqref{eq:Laplaceexp}.
\end{proof}
While Proposition \ref{prop1} shows that the magnetization leads for critical $\beta=1$ to a spectral gap that closes at least like $\sim N^{-1/2}$, when $h=0$,  the next Proposition shows that the magnetization does not imply a vanishing spectral gap when $h > 0.$
\begin{prop}
Let $h>0$, $\beta \ge 1$, and $f$ a function with Fourier transform supported on sets of cardinality $\le k$ for some fixed $k \in \mathbb N_0$ independent of $N$, i.e. 
\[ f = \sum_{S \subset[N]; \left\lvert S \right\rvert \le k }\widehat{f}(S) \chi_S. \] 
Then such functions satisfy an improved inequality with $\varphi_{\text{min}}=\operatorname{argmin}_{\varphi} V_1(\varphi)$  
\begin{equation}
\begin{split}
\operatorname{Var}_{\nu_N}(\mathbb E_{\mu_{\varphi}}(f)) &\le \frac{ \beta^2 \operatorname{csch}^2(2(\beta \varphi_{\text{min}}+h))}{2V''_1(\varphi_{\text{min}})}  \sum_{x \in [N]}\mathbb E_{\rho} \left\lvert \nabla_{\mathbb S^0}^{(x)} f \right\rvert_{\mathbb R^N}^2(1+o(1))
\end{split}
\end{equation}
with a constant $\frac{ \beta^2 \operatorname{csch}^2(2(\beta \varphi_{\text{min}}+h))}{2V''_1(\varphi_{\text{min}})} (1+o(1))$ that strictly bounded away from zero in the limit $N \rightarrow \infty.$ In particular, $V_1''(\varphi_{\text{min}}) >0$ by the discussion in the beginning of Section \ref{sec:SGSMF}.
\end{prop}
\begin{proof}
Using \eqref{eq:varianceformula}, which applies since $V''_1(\varphi_{\text{min}})>0$ by the discussion in Subsection \ref{sec:SGSMF}, we conclude that 
\begin{equation}
\begin{split}
\label{eq:derhier}
\mathbb E_{\mu_{\varphi}}(f)&= \sum_{ S  \subset [N]} \widehat{f}(S)( \tanh(\beta \varphi+h))^{\vert S \vert}
\end{split}
\end{equation}
implies since 
\[ \frac{d}{d\varphi} \tanh(\beta \varphi+h) = \beta \operatorname{sech}^2(\beta \varphi+h)= \beta  \operatorname{csch}(\beta \varphi+h)^2 \tanh(\beta \varphi+h)^2\]
that
\begin{equation}
\begin{split}
&\operatorname{Var}_{\nu_N}(\mathbb E_{\mu_{\varphi}}(f))\\
&=\frac{1}{2N V''_1(\varphi_{\text{min}})}  \left\lvert \sum_{S \subset[N] } \widehat{f}(S)  \beta \vert S \vert \tanh(\beta \varphi_{\text{min}} +h)^{\vert S \vert+1} \operatorname{csch}(\beta \varphi_{\text{min}} +h)^2 \right\rvert^2(1+o(1))\\
&=\frac{\beta^2\tanh(\beta \varphi_{\text{min}} +h)^2\operatorname{csch}(\beta \varphi_{\text{min}} +h)^4}{2N V''_1(\varphi_{\text{min}})}  \left\lvert \sum_{S \subset[N] } \widehat{f}(S)  \tanh(\beta \varphi_{\text{min}} +h)^{\vert S \vert} \langle \indic_{S}, \indic_{[N]} \rangle_{\RR^N} \right\rvert^2(1+o(1))\\
&\le  \frac{2 \beta^2 \operatorname{csch}^2(2(\beta \varphi_{\text{min}}+h))}{V''_1(\varphi_{\text{min}})}   \left\lvert \sum_{S \subset[N] } \widehat{f}(S)  \tanh(\beta \varphi_{\text{min}} +h)^{\vert S \vert}  \indic_{S} \right\rvert^2(1+o(1))\\
&=  \frac{2 \beta^2 \operatorname{csch}^2(2(\beta \varphi_{\text{min}}+h))}{V''_1(\varphi_{\text{min}})}   \left\lvert \mathbb E_{\rho}\sum_{S \subset[N] } \widehat{f}(S)  \chi_S \indic_{S} \right\rvert^2(1+o(1))\\
&\le \frac{2 \beta^2 \operatorname{csch}^2(2(\beta \varphi_{\text{min}}+h))}{V''_1(\varphi_{\text{min}})}    \mathbb E_{\rho} \left\lvert \sum_{S \subset[N] } \widehat{f}(S)  \chi_S \indic_{S} \right\rvert^2(1+o(1))\\
&=  \frac{ \beta^2 \operatorname{csch}^2(2(\beta \varphi_{\text{min}}+h))}{2V''_1(\varphi_{\text{min}})}  \sum_{x \in [N]}\mathbb E_{\rho} \left\lvert \nabla_{\mathbb S^0}^{(x)} f \right\rvert_{\mathbb R^N}^2(1+o(1))
\end{split}
\end{equation}
where we used \eqref{eq:varianceformula} in the first line, $\vert S \vert=\langle \indic_{S},\indic_{[N]} \rangle$ in the second line, Cauchy-Schwarz and $\operatorname{csch}(x)^4\tanh(x)^2=4\operatorname{csch}(2x)^2$ in the third line, \eqref{eq:expectationformula} and \eqref{eq:derhier} in the fourth line, Jensen's inequality in the fifth line and finally \eqref{eq:dirichlet} in the last line.
\end{proof}
\subsection{Critical multi-component systems}
In this subsection, we prove the multi-component part of Theorem \ref{theo:theo2}: 

\medskip

\begin{proof}[Proof of Theorem \ref{theo:theo2}]
The magnetization $M =N^{-1/2} \sum_{x \in [N]} \sigma(x)$ has in the multi-component case always unit Dirichlet norm
\begin{equation}
\label{eq:Dirnorm}
\sum_{x \in [N]}\mathbb E_{\rho} \left\lvert \nabla_{\mathbb S^{n-1}}^{(x)} M \right\rvert^2=\sum_{x \in [N]}\mathbb E_{\rho}(N^{-1}) =1.
\end{equation}
On the other hand, we can explicitly calculate using the derivative of the modified Bessel function of the first kind, $\partial_z I_{\nu}(z) = \frac{\nu}{z} I_{\nu}(z)+I_{\nu+1}(z)$, and \eqref{eq:renpot}, the expectation value $\mathbb E_{\mu_{\varphi}}(\sigma(x))$ that is independent of $x \in[N]$ for $\varphi \neq 0$
\begin{equation}
\begin{split}
\mathbb E_{\mu_{\varphi}}(\sigma(x)) 
&= e^{NV_n(\varphi)} e^{-\frac{\beta}{2}(1+ \Vert \varphi \Vert^2)} \prod_{y \in [N]} \int_{\mathbb S^{n-1}} e^{-\beta \langle \varphi, \sigma(y)\rangle} \sigma(x) \ dS(\sigma(y)) \\
&= \frac{I_{n/2}(\Vert \beta \varphi \Vert)}{I_{n/2-1}(\Vert \beta \varphi \Vert)} \frac{\varphi}{\Vert \varphi \Vert}.
\end{split}
\end{equation}
Taylor expansion at zero then yields 
\[ \left( \mathbb E_{\mu_{\varphi}}(\sigma(x))  \right)^2=\left(\frac{I_{n/2}(\Vert \beta \varphi \Vert)}{I_{n/2-1}(\Vert \beta \varphi \Vert)}\right)^2 = \frac{\Gamma\left(\frac{n}{2}\right)^2 }{4\Gamma\left(1+\frac{n}{2}\right)^2}\Vert \beta \varphi \Vert^2+ \mathcal O(\Vert \beta \varphi \Vert^4).\]
For the renormalized potential we find by Taylor expansion, which we shall already specialize to critical temperatures $\beta=n$, at zero
\[ V_n(\varphi) = \frac{n}{2} + \frac{n^3}{8+4n} \Vert \varphi \Vert^4+ \mathcal O(\Vert \varphi \Vert^5).\]
For the magnetization $M$ \eqref{eq:mag}, we can write
\begin{equation}
\begin{split}
\label{eq:varianceIII}
\operatorname{Var}_{\nu_{N}}(\mathbb E_{\mu_{\varphi}}(M)) &=  \tfrac{1}{N} \sum_{ x,y \in [N]}  \mathbb E_{\nu_N}\left(\left( \frac{I_{n/2}(\Vert \beta \varphi \Vert)}{I_{n/2-1}(\Vert \beta \varphi \Vert)}\right)^{2}\right) \\ &= N \mathbb E_{\nu_N}\left(\left( \frac{I_{n/2}(\Vert \beta \varphi \Vert)}{I_{n/2-1}(\Vert \beta \varphi \Vert)}\right)^{2}\right).
\end{split}
\end{equation}
We then have by radial symmetry of both the renormalized potential and the integrand that at critical temperatures $\beta=n$
\[\mathbb E_{\nu_N}\left(\left( \frac{I_{n/2}(\Vert n \varphi \Vert)}{I_{n/2-1}(\Vert n \varphi \Vert)}\right)^{2}\right)=\frac{\int_0^{\infty} e^{-NV_n(r)} r^{n-1} \left( \frac{I_{n/2}(n r)}{I_{n/2-1}(n r)}\right)^2 dr }{\int_0^{\infty} e^{-NV_n(r)} r^{n-1} \ dr }.\]
Applying Laplace's principle, cf. \cite[Ch.\@ II,Theorem $1$]{W01}, with constants $\mu=4$ and $\alpha=3+(n-1)$ implies that 
\[\mathbb E_{\nu_N}\left(\left( \frac{I_{n/2}(\Vert n \varphi \Vert)}{I_{n/2-1}(\Vert n \varphi \Vert)}\right)^{2}\right) \sim N^{n/4}N^{-(n+2)/4} =N^{-1/2}.\]
Combining this asymptotic behavior with \eqref{eq:Dirnorm} and \eqref{eq:varianceIII} then yields the multi-component claim of Theorem \ref{theo:theo2}, \textit{i.e.} the rate $N^{1/2}$ is caught for the trial (mean spin) function $M$ and thus the spectral gap is decaying at least with speed $N^{-1/2}$. 
\end{proof} 

The following Proposition shows that the upper bound $N^{-1/2}$ on the spectral gap in the critical regime $\beta=n$ for all dimensions $n\ge1$, is in fact sharp:
\begin{prop}
Let $h=0$ and $\beta=n \ge 1.$ The spectral gap of the radial renormalized Schr\"odinger operator grows as $\Theta(N^{1/2})$ and in particular, the spectral gap of the full measure does not close faster than $\Theta( N^{-1/2}).$
\end{prop}
\begin{proof}
Let $\lambda:= N/2$, we then consider the equivalent Schr\"odinger operators to the renormalized generator 
\begin{equation}
\begin{split}
\mathcal H_1&:=-\partial_x^2 + \lambda^2 \vert V'_1(x) \vert^2-\lambda V_1''(x) \text{ and for }n \ge 2 \quad  \\
\mathcal H_n^{\ell}&:=-\left(\partial_r^2 + \tfrac{n-1}{r} \partial_r\right)+\tfrac{\ell(\ell+n-2)}{r^2} +\lambda^2 \vert \nabla V_n \vert^2 -\lambda \Delta  V_n, \quad \ell \in \mathbb N_0,
\end{split}
\end{equation}
where we used that by rotational symmetry of the renormalized potential, for $n \ge 2$, we can decompose the Schr\"odinger operator into individual angular sectors parametrized by $\ell \in \mathbb N_0$. We then introduce auxiliary Schr\"odinger operators
\begin{equation}
\begin{split}
\label{eq:Hamiltonian2}
H_1&=-\partial_{x}^2 + \lambda^2 \tfrac{x^6}{9} - \lambda x^2 \text{ and for } n \ge 2\\
H_n^{\ell}&=-\left(\partial_r^2 + \tfrac{n-1}{r} \partial_r\right) +\tfrac{\ell(\ell+n-2)}{r^2} + \lambda^2\tfrac{n^6}{(2+n)^2} r^6 -\lambda \tfrac{3n^3}{2+n}  r^2, \quad \ell \in \mathbb N_0
\end{split}
\end{equation}
on $L^2(\RR)$ and $L^2((0,\infty), r^{n-1} \ dr)$, respectively. The five first eigenvalues of $H_1$ are shown in Fig. \ref{fig:S1}.
\begin{figure}
\centerline{\includegraphics[height=7cm]{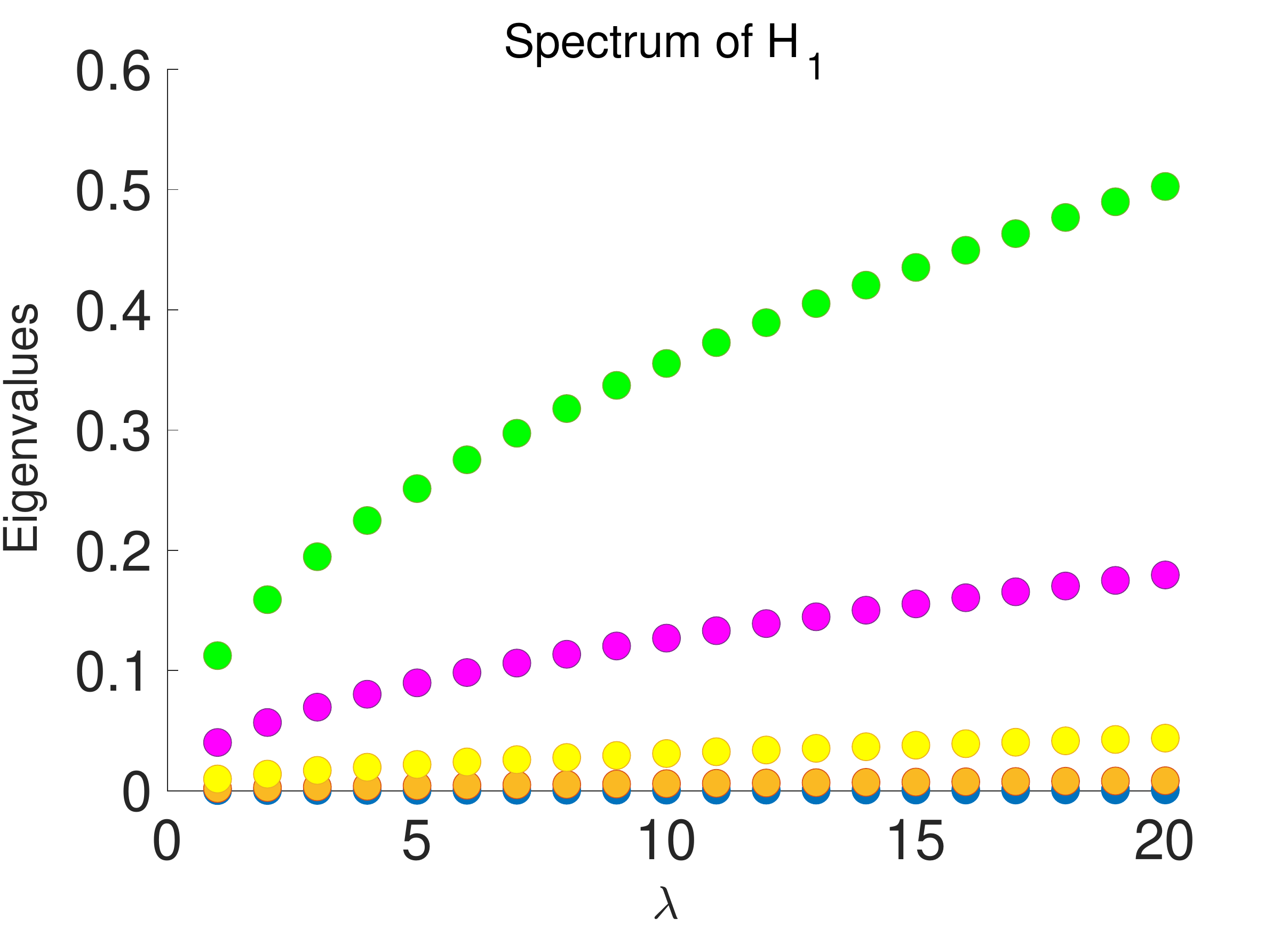}} 
\caption{The five smallest eigenvalues of the operator $H_{1}$ as a function of $\lambda.$ The smallest eigenvalue stays at zero.}
\label{fig:S1}
\end{figure} 
We then define $j \in C_c^{\infty}(-2,2)$ such that $j(x)=1$ for $\left\lvert x \right\rvert \le 1$ and from this 
\begin{equation}
\begin{split}
\label{eq:property3}
J_0(x) &= j(\lambda^{2/9}\left\lvert x \right\rvert )\text{ and }J := \sqrt{1-J_0^2} \text{ with }\left\lVert \nabla  J_{0} \right\rVert_{\mathbb R^n} = \mathcal O(\lambda^{4/9}).
\end{split}
\end{equation}
Invoking the unitary maps $U_{1}  \in \mathcal L(L^2(\mathbb R))$ and $U_n \in \mathcal L( L^2((0,\infty),r^{n-1} \ dr))$ defined as 
\begin{equation}
\begin{split}
 (U_{1}f)(x) := \lambda^{-1/8}f(\lambda^{-1/4}(x)) \text{ and }(U_{n}f)(r) := \lambda^{-n/8}f(\lambda^{-1/4}r)
\end{split}
\end{equation}
shows that the two Schr\"odinger operators in \eqref{eq:Hamiltonian2} are in fact unitarily equivalent, up to multiplication by $\sqrt{\lambda}$, to the $\lambda$-independent Schr\"odinger operators
\begin{equation}
\begin{split}
\label{eq:S2}
S_{1}&= - \partial_x^2+   \tfrac{1}{9}  x^6-x^2 \text{ and for } n \ge 2 \\
S_{n}^{\ell} &= - \left(\partial_r^2  + \tfrac{n-1}{r} \partial_r\right)+\tfrac{\ell(\ell+n-2)}{r^2}+ \tfrac{n^6}{(2+n)^2}  r^6 -\tfrac{3n^3}{2+n}  r^2, \quad \ell \in \mathbb N_0
\end{split}
\end{equation}
respectively. That $\inf(\Spec(S_n^0))=0$ is shown in Section \ref{sec:SUSYQM}. Since $\tfrac{\ell(\ell+n-2)}{r^2}>0$ we have consequently that for $\ell>0$ by monotonicity $\inf(\Spec(S_n^{\ell}))\ge \inf(\Spec(S_n^{1}))>0.$
More precisely, we have that
\begin{equation}
\begin{split}
\label{eq:uniteqv3}
\lambda^{1/2}  U_{n}^{-1}\ S_n^{\ell} U_{n}&= H_n^{\ell}.
\end{split}
\end{equation}
More precisely, since $(U_{n}f)(x) := \lambda^{-n/8}f(\lambda^{-1/4}x),$ it follows that 
\begin{equation}
\begin{split}
(S_n^{\ell} U_{n}f)(r)=&-\lambda^{-n/8} \Bigg(\left(\lambda^{-1/2}f''(\lambda^{-1/4}r)  + \lambda^{-1/4}\tfrac{n-1}{r} f'(\lambda^{-1/4}r)\right)+\tfrac{\ell(\ell+n-2)}{r^2}f(\lambda^{-1/4}r)\\
&+ \tfrac{n^6}{(2+n)^2}  r^6f(\lambda^{-1/4}r) -\tfrac{3n^3}{2+n}  r^2f(\lambda^{-1/4}r)\Bigg).
\end{split}
\end{equation}
Then, applying $(U_n^{-1}f)(r) = \lambda^{n/8} f(\lambda^{1/4}r)$ shows that 
\begin{equation}
\begin{split}
 (U_n^{-1}S_n^{\ell} U_{n}f)(r)=& \lambda^{-1/2}\Bigg(-\left(f''(r)  + \tfrac{n-1}{r} f'(r)\right)+\tfrac{\ell(\ell+n-2)}{r^2}f(r)\\
 &+ \lambda^2 \tfrac{n^6}{(2+n)^2}  r^6f(r) -\lambda \tfrac{3n^3}{2+n}  r^2f(r)\Bigg).
 \end{split}
\end{equation}

Taylor expansion of the potential at $0$ and the estimate on the gradient \eqref{eq:property3} imply that  
\[\left\lvert J_0(\mathcal H_n^{\ell}-H_n^{\ell})J_0 \right\rvert = \mathcal O(\lambda^{4/9}). \]
Let $0=e_1<e_2\le..$ be the eigenvalues (counting multiplicities) of $S_n$ (over all angular sectors $\ell$) and choose $\tau$ such that $e_{n+1}>\tau >e_n$ with $P$ being the projection onto the eigenspace to all eigenvalues of $H$ below $\tau \sqrt{\lambda}.$
The IMS formula, see \cite[(11.37)]{CFKS} for a version on manifolds, implies that  
\[ \mathcal H_n = J\mathcal H_nJ  - \vert \nabla J \vert^2 + \left(J_{0}\mathcal H_nJ_{0}-  \vert \nabla J_{0} \vert^2 \right) \]
such that 
\begin{equation}
\label{eq:usefulrepresent2}
\mathcal H_n = J \mathcal H_n J -\vert \nabla J \vert^2+\left(J_{0} H_n J_{0} + J_{0} (\mathcal H_n-H_n) J_{0}-  \vert \nabla J_{0} \vert^2 \right).
\end{equation}
On the other hand, it follows that 
\begin{equation*} 
\begin{split}
J_{0} H_n J_{0}
&= J_{0} H_n P J_{0} + J_{0} H_n (\operatorname{id}-P) J_{0} \ge J_{0} H_nP_n J_{0} + \sqrt{\lambda} e_n J_{0}^2. 
\end{split}
\end{equation*}
By construction, since $\nabla V_n$ vanishes to third order on the support of $J_0,$ we have 
\[ \Vert \nabla V_n \Vert^2_{\mathbb R^n} \ge c (\lambda^{-2/9})^6 =c\lambda^{-4/3} \text{ on } J\text{ for some }c>0.\] 
Since $\Delta V_n$ vanishes to second order 
\[ \Vert \Delta V_n \Vert_{\mathbb R^n} \ge c \lambda^{-4/9}  \text{ on } J\text{ for some }c>0.\] 
This implies for large $\lambda$ that
\begin{equation}
\label{eq:support3}
 J H J \ge \sqrt{\lambda} e_n J^2.
 \end{equation}
From \eqref{eq:usefulrepresent2} we then conclude that for some $C>0$
\[ \mathcal H_n \ge \sqrt{\lambda} e_n J^2 - C \lambda^{4/9} +J_{0} H_nP J_{0}= \sqrt{\lambda} e_n+ J_{0} H_nP J_{0}-o(\sqrt{\lambda}). \]
This implies the claim of the Proposition, since \[\operatorname{rank} \left(J_{0} H_nPJ_{0} \right) \le n.\] More precisely, for the eigenvalues $E_1(\lambda)\le E_2(\lambda)\le..$ of $\mathcal H_n$ we have shown that 
\[\liminf_{\lambda \rightarrow \infty} \lambda^{-1/2}E_n(\lambda)\ge e_n.\] 
In particular, the lowest possible eigenvalue $e_1=0$ of the renormalized Schr\"odinger operator is of course attained as the nullspace of the renormalized Schr\"odinger operator is non-trivial.
This shows that the spectral gap of the renormalized Schr\"odinger operator grows at least proportional to $\sqrt{\lambda}.$ 
\end{proof}

\begin{appendix}
\label{sec:appendix}
\section{Numerical results}

Recall that the eigenfunctions of the operator
\begin{equation}
\label{eq:Hosc}
H_{\operatorname{osc}}:=-\frac{\hbar^2}{2\mu}\frac{d^2}{dx^2}+\frac{\mu\omega^2}{2} x^2 
\end{equation}
are given for $n \in \mathbb N_0$ by 
\[ \psi_n(x):=\frac{1}{\sqrt{2^n n!} } \left( \frac{\mu\omega}{\pi \hbar } \right)^{1/4} e^{-\frac{ \mu \omega x^2}{2 \hbar }} H_n\left( \sqrt{\frac{\mu \omega}{\hbar}}x \right).\]
Then, it follows that 

\[ \langle \psi_n, -\hbar^2 \psi_m'' \rangle_{L^2(\RR)}= \begin{cases} 
&\frac{\hbar \mu \omega}{2}(2n+1), \text{ if } n=m \\
&-\frac{\hbar \mu \omega}{2}\sqrt{n(n-1)}, \text{ if } n=m+2 \\
&-\frac{\hbar \mu \omega}{2}\sqrt{(n+1)(n+2)}, \text{ if } n=m-2. 
\end{cases}. \]

and 

\[ \langle \psi_n,x^2 \psi_m \rangle_{L^2(\RR)}= \begin{cases} 
&\frac{\hbar}{2 \mu \omega}(2n+1), \text{ if } n=m \\
&\frac{\hbar}{2 \mu \omega}\sqrt{n(n-1)}, \text{ if } n=m+2 \\
&\frac{\hbar}{2 \mu \omega}\sqrt{(n+1)(n+2)}, \text{ if } n=m-2. 
\end{cases}. \]

Using the annihilation operator $a=2^{-1/2}(\partial_q+q)$ where $q=\sqrt{\frac{\mu \omega}{\hbar}} x$ and its adjoint we can explicitly compute the matrix elements of all $(\langle \psi_n, x^n \psi_m \rangle)$ by writing $q^n=\sqrt{2}(a+a^*)^n$ and using the well-known action of the annihilation operator on eigenstates of \eqref{eq:Hosc}. Using a finite-basis truncation of the above matrices allowed us then to obtain Figures \ref{fig:Sop} and \ref{fig:S1}.

\section{Asymptotic properties of the Ising model}
\begin{lemm}
\label{lem:asym}
Let $\beta>1$ and $h \in [0,h_{\text{c}}).$ The three critical points of $\eta_{N,h}:(-1,1)\rightarrow \mathbb R$
\[ \eta_{N,h}(s) =\frac{\Gamma(N+1)}{\Gamma(N/2(1+s)+1)\Gamma(N/2(1-s)+1)} e^{-\tfrac{N\beta}{2}(1-s^2)+Nhs}\]
are given by $s_{N}^{\text{c}} := \gamma(\beta)  (1+\smallO(1))$ where $\gamma(\beta)$ satisfies the critical equation for the continuous renormalized potential
\[ \gamma(\beta) = \tanh(\beta \gamma(\beta)+h).\]
Let us order the solutions $\gamma(\beta)$ to that equation by $\gamma_1(\beta)<\gamma_2(\beta)<\gamma_3(\beta).$
For $h=0$ the function $\eta_{N,0}$ attains (in the limit $N \rightarrow \infty$) its maximum at $\gamma_1(\beta)=-\gamma_3(\beta)<0$ and minimum at $\gamma_2(\beta)=0$. 

\medskip

Let $h>0$, then the function $\eta_{N,h}$ attains (in the limit $N \rightarrow \infty$) its unique global maximum at $\gamma_3(\beta)>0$ whereas both $\gamma_1(\beta),\gamma_2(\beta)<0$ and $ \gamma_1(\beta),\gamma_2(\beta)$ are local maxima and minima respectively.

\medskip

The logarithmic derivative $\zeta_{N,h}(s) = \partial_s \log(\eta_{N,h}(s))$ satisfies 
\begin{equation}
\label{eq:logder}
\zeta_{N,h}(s)=    N\left(\beta s-\operatorname{arctanh}\left(s \right)+  h\right)(1+\smallO(1)).
\end{equation}
\end{lemm}
\begin{proof}
For $h=0$ we note that $\eta_{N,0}$ is even and for $h>0$ the global maxima of $\eta_{N,h}$ must be attained at some $s \ge 0.$
Direct computations show by the logarithmic scaling of the digamma function $\psi_2(s) = \log(\Gamma)'(s)= \log(s)+ \mathcal O(1/s)$ that the logarithmic derivative $\zeta_{N,h}(s) = \partial_s \log(\eta_{N,h}(s))$ is given by \eqref{eq:logder}.
Thus, for all critical values $s_N^{\text{c}}$ of $\eta_{N,h}$, i.e. those values that satisfy $\zeta_{N,h}(s_N^{\text{c}})=0,$ there exists $\gamma(\beta) \in [-1,1]$ such that $\gamma(\beta):=  \lim_{N \rightarrow \infty}s_N^{\text{c}}$
and $\gamma(\beta)$ is any solution to $ \gamma(\beta) = \operatorname{tanh}(\gamma(\beta) \beta+ h).$

\medskip

We then obtain \eqref{eq:logder} directly by differentiating $\log(\eta_{N,h})$ and using the identity
\begin{equation}
\begin{split}
&-\partial_s \log\left(\Gamma\left(\tfrac{N(1+s)}{2}+1\right)\Gamma\left(\tfrac{N(1-s)}{2}+1\right)\right) \\
&= -\frac{N}{2} \left(\log\left(\frac{1+N/2(1+s)}{1+N/2(1-s)}\right) \right)(1+\smallO(1)) \\
 &= -\frac{N}{2}\left(\log\left( \frac{1 + s \frac{N/2}{1+N/2}}{1 - s\frac{N/2}{1+N/2}} \right) \right)(1+\smallO(1)) \\
 &= -\frac{N}{2}\left(\log\left( \frac{1 + s }{1 - s} \right) \right)(1+\smallO(1)) \\
  &= -N \operatorname{artanh}(s)(1+\smallO(1)).
 \end{split}
 \end{equation}

\medskip

Moreover, we read off from \eqref{eq:logder} that 
\[\lim_{k \uparrow 1} \zeta_{N,h}(k) = -\infty\text{ and }\lim_{k \downarrow -1}  \zeta_{N,h}(k) = \infty.\]
In particular, $\gamma(\beta)$ solves the implicit equation $\gamma(\beta) = \tanh(\beta \gamma(\beta)+h).$ For the second derivative of $\zeta_{N,h}$ which is $h$-independent, we find the closed-form expression using the derivative of the trigamma function $\psi_3$
\begin{equation}
\begin{split}
\zeta_{N,h}''(s) &= -\frac{N^3}{8} \left(\psi_{3}'(1+N/2(1+s))-\psi_{3}'(1+N/2(1-s))\right) \\
&= \frac{N^3}{8} \int_0^{\infty} z^2 \frac{e^{-z(1+N/2(1+s))}-e^{-z(1+N/2(1-s))}}{1-e^{-z}} \ dz.
 \end{split}
 \end{equation}
This implies that $\zeta_N$ is strictly convex on $[-1,0)$ and strictly concave on $(0,1]$. Using the asymptotic of the trigamma function $\psi_3(s)= 1/s + 1/(2s^2)+ \mathcal O(1/s^3)$ we find that $\zeta_{N,h}$ is strictly monotone increasing at zero, independent of $h$,
\[ \zeta_N'(0) =N\beta -  \frac{N^2 \psi_3(1+N/2)}{2} = N\left(\beta -   \frac{N/2}{1+N/2}\right)(1+\smallO(1))>0, \]
since $\beta>1.$
\end{proof}
\begin{lemm}
\label{eq:auxlemma}
The function $\mathcal I(r):= r \frac{I_{n/2-1}(r)}{I_{n}(r)}$ is strictly monotonically increasing on $(0,\infty).$ In particular $\mathcal I'(r)>0$.
\end{lemm}
\begin{proof}
By differentiating and using that $I_{\nu}'(r) = \frac{\nu}{r} I_{\nu}(r) + I_{\nu+1}(r)$ we find
\[\mathcal I'(r) = r\left(1-\frac{I_{n/2-1}(r)I_{n/2+1}(r)}{I_n(r)^2} \right).\]
Thus, it suffices to record that the product of Bessel functions satisfies $I_{n/2}(z)^2 > (I_{n/2-1}I_{n/2+1})(z):$
\begin{equation}
\begin{split}
(I_{n/2-1}I_{n/2+1})(z) &= \left(z/2 \right)^{n} \sum_{k=0}^{\infty} \frac{(n+k+1)_{k}(z^2/4)^k}{k! \Gamma(n/2-1+k+1) \Gamma(n/2+1+k+1)} \\
(I_{n/2}I_{n/2})(z) &= \left(z/2 \right)^{n} \sum_{k=0}^{\infty} \frac{(n+k+1)_{k}(z^2/4)^k}{k! \Gamma(n/2+k+1) \Gamma(n/2+k+1)} \\
\end{split}
\end{equation}
Hence, the identity follows from 
\[\Gamma(n/2+k+1)^2<\Gamma(n/2-1+k+1) \Gamma(n/2+1+k+1)\] 
which follows itself from logarithmic convexity of the gamma function.
\end{proof}
\section{SUSY Quantum Mechanics}
\label{sec:SUSYQM}
We use ideas from \emph{supersymmetric quantum mechanics}, to show positivity and analyze the ground state of several Schr\"odinger operators appearing in this article:

In one dimension, we recall that using operators 
\[ A=\partial_x+W(x) \text{ and } A^* = -\partial_x+W(x)\]
with real-valued and smooth \emph{superpotential} $W$,
we can write 
\[ A^*A = -\partial_x^2 -W'(x)+W(x)^2 \text{ and } AA^*=-\partial_x^2+W'(x)+W(x)^2.\]
In particular, $W(x):=\sqrt{\beta(\beta-1)}\beta x^2$ yields operator $S_{\varphi_{\pm}}$ defined in \eqref{eq:S}.

However, solving $A\psi=0$ or $A^*\psi = 0 $ shows that $\psi=C e^{\pm \sqrt{\beta(\beta-1)}\beta x^3}\notin L^2(\RR).$ This shows that $\inf(\Spec(AA^*)),\inf(\Spec(A^*A))>0.$

We now analyze operators in \eqref{eq:S2}.
Choosing $W(x):=\frac{x^3}{3}$, yields $A^*A=S_1$ in \eqref{eq:S2}, and we find by solving 
$A\psi(x)=0$ that $\psi(x)\propto e^{-x^4/12}$ which implies that $\inf(\Spec(A^*A))=0.$

For radial operators on $L^2((0,\infty),r^{n-1} \ dr)$, a similar argument applies: 

We define operators 
\[ A=\partial_r + W(r) \text{ and } A^*= -\partial_r +\frac{n-1}{r} + W(r).\]

Choosing then $W(r):=\frac{n^3}{(2+n)}r^3,$ such that $A^*A =S_n$ with $S_n$ in \eqref{eq:S2}, we find by solving 
\[ A\psi=0 \Rightarrow \psi(r)\propto e^{-\frac{n^3}{4(n+2)} r^4} \in L^2((0,\infty),r^{n-1} \ dr).\]

\section{Asymptotic properties}

\begin{lemm}\cite[Theo $1.4.10$]{BBS18}
\label{aoev}
Let $V: \mathbb R \rightarrow \mathbb R$ be smooth with unique global minimum at $\varphi_{\text{min}} \in \mathbb R$ and $V''(\varphi_{\text{min}})>0.$ Assume that $\int_{\mathbb R} e^{-V(\varphi)} \ d\varphi$ is finite and that $\left\{ \varphi \in \mathbb R; V(\varphi) \le V(\varphi_{\text{min}})+1 \right\}$ is compact. We also define the probability measure $d\zeta_N(\varphi) = e^{-N V(\varphi)} \ d\varphi / \int_{\mathbb R}e^{-N V(\varphi)} d \varphi.$ Then for any bounded smooth function $g: \mathbb R \rightarrow \mathbb R$ 
\begin{equation}
\begin{split}
\label{eq:expectationformula}
\mathbb E_{\zeta_N}(g)&= \frac{\int_{\mathbb R} g(\varphi)e^{-NV(\varphi)} d \varphi}{\int_{\mathbb R} e^{-NV(\varphi)} d \varphi}\\
& = g(\varphi_{\text{min}})+ \frac{g''(\varphi_{\text{min}})}{2N V''(\varphi_{\text{min}})} +\frac{3 V'''(\varphi_{\text{min}}) g'(\varphi_{\text{min}})}{4N V''(\varphi_{\text{min}})^3} + \mathcal O(1/N^2)
\end{split}
\end{equation}
and for the variance
\begin{equation}
\begin{split}
\label{eq:varianceformula}
\operatorname{Var}_{\zeta_N} (g) = \frac{g'(\varphi_{\text{min}})^2}{NV''(\varphi_{\text{min}})} + \mathcal O(1/N^2)
\end{split}
\end{equation}
\end{lemm}

\end{appendix}


\bibliographystyle{alpha}
\bibliography{bibliography}

\end{document}